\documentclass{llncs}

\usepackage{t1enc}
\usepackage{amsmath}
\usepackage{latexsym}
\usepackage{theorem}
\usepackage{amssymb}
\usepackage{url}

\newcommand{\version}[2]{#1}
\newcommand{\comment}[1]{}

\newcommand{\hsp}[1][3ex]{\hspace*{#1}}
\newcommand{\vsp}[1][1mm]{\vspace*{#1}}
\newcommand{\moins}{\setminus}
\newcommand{\vide}{\emptyset}
\newcommand{\eg}{{\em e.g.} }
\newcommand{\ie}{{\em i.e.} }

\newcommand{\dom}{\mr{dom}}

\renewcommand{\a}{\rightarrow}
\newcommand{\A}{\Rightarrow}

\renewcommand{\to}{\mapsto}

\newcommand{\ab}{\a_\b}

\newcommand{\ar}{\a_\cR}

\newcommand{\ps}[1]{{\langle #1\rangle}}

\newcommand{\I}[1]{[\![#1]\!]}

\newcommand{\ex}{\exists}
\newcommand{\all}{\forall}

\renewcommand{\th}{\vdash}

\newcommand{\sle}{\subseteq}

\renewcommand{\o}[1]{{\overline{#1}}}
\renewcommand{\u}[1]{{\underline{#1}}}

\renewcommand{\b}{\beta}

\newcommand{\G}{\Gamma}
\renewcommand{\d}{\delta}
\newcommand{\D}{\Delta}
\newcommand{\ep}{\epsilon}
\newcommand{\vep}{\varepsilon}
\newcommand{\z}{\zeta}
\renewcommand{\t}{\theta}
\newcommand{\T}{\Theta}
\newcommand{\io}{\iota}

\renewcommand{\l}{\lambda}
\renewcommand{\L}{\Lambda}
\renewcommand{\r}{\rho}
\newcommand{\s}{\sigma}
\renewcommand{\S}{\Sigma}

\newcommand{\vphi}{\varphi}
\newcommand{\w}{\omega}

\newcommand{\mc}{\mathcal}

\newcommand{\mr}{\mathrm}
\newcommand{\mb}{\mathbb}

\newcommand{\mk}{\mathfrak}
\newcommand{\mf}{\mathsf}

\newcommand{\bB}{\mb{B}}

\newcommand{\bE}{\mb{E}}
\newcommand{\bF}{\mb{F}}

\newcommand{\bM}{\mb{M}}
\newcommand{\bN}{\mb{N}}

\newcommand{\bT}{\mb{T}}

\newcommand{\cA}{\mc{A}}
\newcommand{\cB}{\mc{B}}

\newcommand{\cF}{\mc{F}}

\newcommand{\cM}{\mc{M}}
\newcommand{\cN}{\mc{N}}

\newcommand{\cR}{\mc{R}}
\newcommand{\cS}{\mc{S}}
\newcommand{\cT}{\mc{T}}

\newcommand{\cX}{\mc{X}}

\newcommand{\cZ}{\mc{Z}}

\newcommand{\ka}{\mk{a}}
\newcommand{\kb}{\mk{b}}

\newcommand{\kt}{\mk{t}}
\newcommand{\ku}{\mk{u}}

\newcommand{\kA}{\mk{A}}

\newcommand{\fb}{\mf{b}}
\newcommand{\fc}{\mf{c}}
\newcommand{\fd}{\mf{d}}
\newcommand{\fe}{\mf{e}}
\newcommand{\ff}{\mf{f}}
\newcommand{\fg}{\mf{g}}
\newcommand{\fh}{\mf{h}}

\newcommand{\fj}{\mf{j}}

\newcommand{\fs}{\mf{s}}

\newcommand{\fA}{\mf{A}}
\newcommand{\fB}{\mf{B}}
\newcommand{\fC}{\mf{C}}

\newcommand{\fF}{\mf{F}}

\newcommand{\fN}{\mf{N}}
\newcommand{\fO}{\mf{O}}

\newcommand{\fT}{\mf{T}}

\newcommand{\va}{{\vec{a}}}
\newcommand{\vb}{{\vec{b}}}

\newcommand{\vl}{{\vec{l}}}
\newcommand{\vm}{{\vec{m}}}

\newcommand{\vp}{{\vec{p}}}

\newcommand{\vt}{{\vec{t}}}
\newcommand{\vu}{{\vec{u}}}
\newcommand{\vv}{{\vec{v}}}

\newcommand{\vx}{{\vec{x}}}
\newcommand{\vy}{{\vec{y}}}
\newcommand{\vz}{{\vec{z}}}

\newcommand{\vB}{{\vec{B}}}

\newcommand{\vP}{{\vec{P}}}

\newcommand{\vT}{{\vec{T}}}
\newcommand{\vU}{{\vec{U}}}
\newcommand{\vV}{{\vec{V}}}

\newenvironment{rul}
  {$\begin{array}{rcll}}
  {\end{array}$}

\newenvironment{rew}[1][~~\a~~]
  {$\begin{array}{r@{#1}ll}}
  {\end{array}$}
\newenvironment{rewc}[1][~~\a~~]
  {\begin{center}\begin{rew}[#1]}
  {\end{rew}\end{center}}

{\theorembodyfont{\rmfamily} 

}

\newenvironment{lstgeneric}[2]
  {\begin{list}{#1}{\topsep=.5mm\itemsep=.5mm\parsep=0mm%
    \itemindent=-3ex\labelsep=1ex\labelwidth=0ex #2}}
  {\end{list}}

\newcommand{\SN}{\cS\cN}

\newcommand{\pos}{\mr{Pos}}
\newcommand{\ind}{\mr{Ind}}
\newcommand{\mon}{\mr{Mon}}

\newcommand{\lx}{\l x}

\newcommand{\les}{\le_\cA}
\newcommand{\ges}{\ge_\cA}
\newcommand{\lts}{<_\cA}
\newcommand{\gts}{>_\cA}

\newcommand{\ltf}{<_\cF}
\newcommand{\gtf}{>_\cF}
\newcommand{\eqf}{\simeq_\cF}

\newcommand{\geb}{\le_\cB}
\newcommand{\ltb}{<_\cB}

\newcommand{\eqb}{\simeq_\cB}

\newcommand{\valpha}{{\vec\alpha}}
\newcommand{\vbeta}{{\vec\b}}

\newcommand{\thc}{\th_{\ff\va}}

\newcommand{\tf}{\tau_\ff}

\newcommand{\tg}{\tau_\fg}
\newcommand{\tc}{\tau_\fc}

\renewcommand{\max}{\mf{max}}
\renewcommand{\lim}{\mf{lim}}
\newcommand{\rec}{\mf{rec}}

\newcommand{\Set}{\mbox{\bf Set}}

\begin{document}

\title{\bf On the relation between sized-types based termination and semantic labelling}

\author{Fr\'ed\'eric Blanqui$^1$ \and Cody Roux$^2$ (INRIA)}

\institute{
FIT 3-604, Tsinghua University, Haidian District, Beijing 100084, China,
{\tt frederic.blanqui@inria.fr}\\
\and
LORIA\footnote{UMR 7503 CNRS-INPL-INRIA-Nancy2-UHP},
Pareo team, Campus Scientifique,
BP 239, 54506 Vandoeuvre-l\`es-Nancy, Cedex, France,
{\tt cody.roux@loria.fr}}

\maketitle

\begin{abstract}
We investigate the relationship between two independently developed
termination techniques. On the one hand, sized-types based termination
(SBT) uses types annotated with size expressions and Girard's
reducibility candidates, and applies on systems using constructor
matching only. On the other hand, semantic labelling transforms a
rewrite system by annotating each function symbol with the semantics
of its arguments, and applies to any rewrite system.

\hsp
First, we introduce a simplified version of SBT for the simply-typed
lambda-calculus. Then, we give new proofs of the correctness of SBT
using semantic labelling, both in the first and in the higher-order
case. As a consequence, we show that SBT can be extended to systems
using matching on defined symbols (\eg associative functions).
\end{abstract}

\section{Introduction}

Sized types were independently introduced by Hughes, Pareto and Sabry
{}\cite{hughes96popl} and Gim\'enez {}\cite{gimenez96phd}, and were
extended to richer type systems, to rewriting and to richer size
annotations by various researchers
{}\cite{xi01lics,abel08lmcs,barthe04mscs,blanqui04rta,blanqui06lpar-cbt}.

Sized types are types annotated with size expressions. For instance,
if $\fT$ is the type of binary trees then, for each $a\in\bN$, a type
$\fT^a$ is introduced to type the trees of height smaller or equal to
$a$. In the general case, the size is some ordinal related to the
interpretation of types in Girard's reducibility candidates
{}\cite{girard72phd}. However, as suggested in {}\cite{blanqui04rta},
other notions of sizes may be interesting.

These size annotations can then be used to prove the termination of
functions by checking that the size of arguments decreases along
recursive calls, but this applies to functions defined by using
matching on constructor terms only.

At about the same time, semantic labelling was introduced for
first-order systems by Zantema {}\cite{zantema95fi}. It received a lot
of attention in the last years
and was recently extended to the higher-order case by Hamana
{}\cite{hamana07ppdp}.

In contrast with SBT, semantic labelling is not a termination
criterion but transforms a system into another one whose termination
is equivalent and hopefully simpler to prove. The transformation
consists in annotating function symbols with the semantics of their
arguments in some model of the rewrite system. Finding a model may of
course be difficult. We will see that the notion of size used in SBT
provides such a model.

In this paper, we study the relationship between these two methods. In
particular, we give a new proof of the correctness of SBT using
semantic labelling. This will enable us to extend SBT to systems using
matching on defined symbols.


{\bf Outline.} Section {}\ref{sec-prelim} introduces our
notations. Section {}\ref{sec-sbt} explains what SBT is and Section
{}\ref{sec-sup} introduces a simplified version of it. To ease the
understanding of the paper, we first present the first-order case
which already contains the main ideas, and then consider the
higher-order case which requires more knowledge. Hence, in Section
{}\ref{sec-fosl} (resp. {}\ref{sec-hosl}), we recall what is semantic
labelling in the first (resp. higher) order case and show in Section
{}\ref{sec-fo} (resp. {}\ref{sec-ho}) that SBT is an instance of
it. For lack of space, some proofs are given in the Appendices of
{}\cite{blanqui09cslfull}.

\section{Preliminaries}
\label{sec-prelim}



{\bf First-order terms.} A {\em signature} $\cF$ is made of a set
$\cF_n$ of {\em function symbols} of {\em arity} $n$ for each
$n\in\bN$. Let $\cF$ be the set of all function symbols. Given a set
$\cX$ of {\em variables}, the set of {\em first-order terms}
$\cT(\cF,\cX)$ is defined as usual: $\cX\sle\cT$; if $\ff\in\cF_n$ and
$\vt$ is a sequence $t_1,\ldots,t_n\in\cT$ of length $n=|\vt|$, then
$\ff(\vt)\in\cT$.

An {\em $\cF$-algebra} $\cM$ is given by a set $M$ and, for each
symbol $\ff\in\cF_n$, a function $\ff^\cM:M^n\a M$. Given a valuation
$\mu:\cX\a M$, the interpretation of a term $t$ is defined as follows:
$\I{x}\mu=\mu(x)$ and
$\I{\ff(t_1,\ldots,t_n)}\mu=\ff^\cM(\I{t_1}\mu,\ldots,\I{t_n}\mu)$.

{\em Positions} are {\em words} on $\bN$. We denote by $\vep$ the {\em
empty} word and by $p\cdot q$ or $pq$ the concatenation of $p$ and
$q$. Given a term $t$, we denote by $t|_p$ the subterm of $t$ at
position $p$, and by $t[u]_p$ the replacement of this subterm by
$u$. Let $\pos(\ff,t)$ be the set of the positions of the occurrences
of $\ff$ in $t$.

\smallskip

{\bf Higher-order terms.} The set of (simple) {\em types} is
$\bT=\cT(\S)$ where $\S_0=\cB$ is a set of {\em base types},
$\S_2=\{\A\}$ and $\S_n=\vide$ otherwise. The sets of {\em positive}
and {\em negative positions in a type} are inductively defined as
follows:

\begin{itemize}
\item $\pos^+(\fB)=\vep$ and $\pos^-(\fB)=\vide$ for each $\fB\in\cB$,
\item $\pos^\d(T\A U)=1\cdot\pos^{-\d}(T)\cup2\cdot\pos^\d(U)$
where $--=+$ and $-+=-$.
\end{itemize}

Let $\cX$ be an infinite set of variables. A typing environment $\G$
is a map from a finite subset of $\cX$ to $\bT$. For each type $T$, we
assume given a set $\cF_T$ of {\em function symbols of type $T$}. The
sets $\L_T(\G)$ of {\em terms of type $T$ in $\G$} are defined as
usual: $\cF_T\sle\L_T(\G)$; if $(x,T)\in\G$ then $x\in\L_T(\G)$; if
$t\in\L_U(\G,x:T)$, then $\lx^Tt\in\L_{T\A U}(\G)$; if $t\in\L_{U\A
  V}(\G)$ and $u\in\L_U(\G)$, then $tu\in\L_V(\G)$.

Let $\cF$ (resp. $\L$) be the set of all function symbols
(resp. terms). \comment{As usual, terms are considered up to renaming
  of bound variables.} Let $\cX(t)$ be the set of {\em free variables}
of $t$. A {\em substitution} $\s$ is a map from a finite subset of
$\cX$ to $\L$. We denote by $(_x^u)$ the substitution mapping $x$ to
$u$, and by $t\s$ the application of $\s$ to $t$. A term $t$ {\em
  $\b$-rewrites} to a term $u$, written $t\ab u$, if there is
$p\in\pos(t)$ such that $t|_p=(\l x^Tv)w$ and $u=t[v_x^w]_p$.

A {\em rewrite rule} is a pair of terms $l\a r$ {\em of the same type}
such that $\cX(r)\sle\cX(l)$. A {\em rewrite system} is a set $\cR$ of
rewrite rules. A term $t$ {\em rewrites} to a term $u$, written $t\ar
u$, if there is $p\in\pos(t)$, $l\a r\in\cR$ and $\s$ such that
$t|_p=l\s$ and $u=t[r\s]_p$.

\smallskip

{\bf Constructor systems.} A function symbol $\ff$ is either a {\em
  constructor symbol} if no rule left-hand side is headed by $\ff$, or
a {\em defined symbol} otherwise. A {\em pattern} is a variable or a
term of the form $\fc\vt$ with $\fc$ a constructor symbol and $\vt$
patterns. A rewrite system is {\em constructor} if every rule is of
the form $\ff\vl\a r$ with $\vl$ patterns.

As usual, we assume that constructors form a valid inductive structure
{}\cite{blanqui05mscs}, that is, there is a well-founded
quasi-ordering $\geb$ on $\cB$ such that, for each base type $\fB$,
constructor $\fc:\vT\A\fB$ and base type $\fC$ occuring at position
$p$ in $T_i$, either $\fC\ltb\fB$ or $\fC\eqb\fB$ and
$p\in\pos^+(T_i)$. Mendler indeed showed that invalid inductive
structures lead to non-termination {}\cite{mendler87phd}.

Given a constructor $\fc:\vT\A\fB$, let $\ind(c)$ be the set of
integers $i$ such that $T_i$ contains a base type $\fC\eqb\fB$. A
constructor $\fc$ with $\ind(\fc)\neq\vide$ is said {\em recursive}.

A constructor $\fc:\vT\A\fB$ is {\em strictly-positive} if, for each
$i$, either no base type equivalent to $\fB$ occurs in $T_i$, or $T_i$
is of the form $\vU\A\fC$ with $\fC\eqb\fB$ and no base type
equivalent to $\fB$ occuring in $\vU$.

SBT applies to constructor systems only. By using semantic labelling,
we will prove that it can also be applied to some non-constructor
systems.

\section{Sized-types based termination}
\label{sec-sbt}

We now present a simplified version of the termination criterion
introduced in {}\cite{blanqui04rta}, where the first author considers
rewrite systems on terms of the Calculus of Algebraic Constructions, a
complex type system with polymorphic and dependent types. Here, we
restrict our attention to simply-typed $\l$-terms since there is no
extension of semantic labelling to polymorphic and dependent types
yet.

This termination criterion is based on the semantics of types in
reducibility candidates {}\cite{girard72phd}. An arrow type $T\A U$ is
interpreted by the set $\I{T\A U}=\{v\in\cT\mid\all
t\in\I{T},vt\in\I{U}\}$. A base type $\fB$ is interpreted by the
fixpoint $\I\fB$ of the monotonic function $F_\fB(X)=\{v\in\SN\mid\all$
constructor $\fc:\vT\A\fB,\all\vt,\all i\in\ind(\fc),v\a^*\fc\vt\A
t_i\in\I{T_i}_{\fB\to X}\}$ on the lattice of reducibility candidates
that is complete for set inclusion {}\cite{blanqui05mscs}. This
fixpoint, defined by induction on the well-founded quasi-ordering
$\geb$ on base types, can be reached by transfinite iteration of
$F_\fB$ up to some limit ordinal $\w_\fB$ strictly smaller than the
first uncountable ordinal $\kA$. This provides us with the following
notion of size: the size of a term $t\in\I\fB$ is the smallest ordinal
$o_\fB(t)=\ka<\kA$ such that $t\in F_\fB^\ka(\bot)$, where $\bot$ is
the smallest element of the lattice and $F_\fB^\ka$ is the function
obtained after $\ka$ transfinite iterations of $F_\fB$.

\comment{
\begin{itemize}
\item $F_\fB^0(x)=x$,
\item $F_\fB^{\ka+1}(x)= F_\fB(F_\fB^\ka(x))$,
\item $F_\fB^\l(x) = \bigcup\{F_\fB^\ka(x)\mid\ka<\l\}$
if $\l$ is a limit ordinal.
\end{itemize}
}

This notion of size, which corresponds to the tree height for
patterns, has the following properties: it is well-founded; the size
of a pattern is strictly bigger than the size of its subterms; if $t\a
t'$ then the size of $t'$ is smaller than (since $\a$ may be non
confluent) or equal to the size of $t$.

SBT consists then in providing a way to syntactically represent the
sizes of terms and, given for each function symbol an annotation
describing how the size of its output is related to the sizes of its
inputs, check that some measure on the sizes of its arguments
decreases in each recursive call.


\smallskip

{\bf Size algebra.} Sizes are represented and compared by using a
first-order term algebra $\cA=\cT(\S,\cX)$ equipped with an ordering
$\les$ such that:

\begin{itemize}
\item
$\lts$ is stable by substitution;
\item
$(\kA,<_\kA)$, where $<_\kA$ is the usual ordering on ordinals, is a
model of $(\cA,\lts)$:
\begin{itemize}
\item
every symbol $\fh\in\S_n$ is interpreted by a function
$\fh^\kA:\kA^n\a\kA$;
\item
if $a\lts b$ then $\I{a}\mu<_\kA\I{b}\mu$ for each $\mu:\cX\a\kA$.
\end{itemize}
\end{itemize}

To denote a size that cannot be expressed in $\cA$ (or a size that we
do not care about), $\S$ is extended with a (biggest) nullary element
$\infty$. Let $\o\cA$ be the extended term algebra in which all terms
containing $\infty$ are identified,
${<_\o\cA}={{<_\cA}\cup{\{(a,\infty)\mid a\in\cA\}}}$ and
${\le_\o\cA}={{\le_\cA}\cup{\{(a,\infty)\mid a\in\o\cA\}}}$. Note that
such an extension is often used in domain theory but with a least
element instead.


\smallskip

{\bf Annotated types.} The set of base types is now all the
expressions $\fB^a$ such that $\fB\in\cB$ and $a\in\o\cA$. The
interpretation of $\fB^\infty$ (also written $\fB$) is $\I\fB$ and,
given $a\in\cA$, the interpretation of $\fB^a$ wrt a size valuation
$\mu:\cX\a\kA$ is the set of terms in $\I\fB$ whose size is smaller or
equal to $\I{a}\mu$: $\I{\fB^a}^\mu=F_\fB^{\I{a}\mu}(\bot)$.

Hence, we assume that every symbol $\ff\in\cF$ is given an annotated
type $\tf^\cA$ whose size variables, like type variables in ML, are
implicitly universally quantified and can be instantiated by any size
expression. Hence the typing rule for symbols in Figure {}\ref{fig-th}
allows any size substitution $\vphi$ to be applied to
$\tf^\cA$. Subtyping naturally follows from the interpretation of
types and the ordering on $\cA$.

\begin{figure}
\vsp[-5mm]
\centering
\caption{Type system with size annotations\label{fig-th}}
\vsp[3mm]
$\cfrac{\vphi:\cX\a\cA}{\G\th^s\ff:\tf^\cA\vphi}$\quad\quad
$\cfrac{(x,T)\in\G}{\G\th^sx:T}$\quad\quad
$\cfrac{\G,x:T\th^su:U\quad x\notin\G}{\G\th^s\l x^Tu:T\A U}$\\[3mm]

$\cfrac{\G\th^st:U\A V\quad \G\th^su:U}{\G\th^stu:V}$\quad\quad
$\cfrac{\G\th^st:T\quad T\le T'}{\G\th^st:T'}$\\[3mm]

$\cfrac{a\le_\o\cA b}{\fB^a\le\fB^b}$\quad\quad
$\cfrac{T'\le T\quad U\le U'}{T\A U\le T'\A U'}$\quad\quad
$\cfrac{T\le U\quad U\le V}{T\le V}$
\end{figure}

\begin{definition}
Given a type $T$, let $T^\infty$ be the type obtained by annotating
every base type with $\infty$, and $annot_\fB^\alpha(T)$ be the type
obtained by annotating every base type $\fC\eqb\fB$ with $\alpha$, and
every base type $\fC\not\eqb\fB$ with $\infty$. Conversely, given an
annotated type $T$, let $|T|$ be the type obtained by removing all
annotations.
\end{definition}

Note that, in constrast to types, terms are unchanged: in $\l x^Tu$,
$T=T^\infty$.

Given a size symbol $\fh\in\S$, let $\mon^+(\fh)$
(resp. $\mon^-(\fh)$) be the sets of integers $i$ such that $\fh$ is
monotonic (resp. anti-monotonic) in its $i$-th argument. The sets of
{\em positive} and {\em negative} positions in an annotated type are:
\comment{$\fB^a$ and a size expression $\fh(\vec\va)$ are then defined as
follows:}

\begin{itemize}
\item $\pos^-(\fB^a)=0\cdot\pos^-(a)$
and $\pos^+(\fB^a)=\{\vep\}\cup 0\cdot\pos^+(a)$,
\item $\pos^-(\alpha)\!=\!\vide,\,\pos^+(\alpha)\!=\!\vep,\,
\pos^\d(\fh(\va))
\!=\!\bigcup\{i\!\cdot\!\pos^{\ep\d}(a_i)\,|\,i\!\in\!\mon^\ep(h),\ep\!\in^{\{-,+\}}\}$.
\end{itemize}

To ease the expression of termination conditions, for every defined
symbol $\ff$, $\tf^\cA$ is assumed to be of the form
$\vP\A\vec\fB^{\valpha_\ff}\A\fB^{\ff^\cA(\valpha_\ff)}$ with
$|\tf^\cA|=\tf$, $\cX(\vP)=\vide$ and
$\cX(\ff^\cA(\valpha_\ff))\sle\{\valpha_\ff\}$ where $\valpha_\ff$ are
pairwise distinct variables. The arguments of type $\vec\fB$ are the
ones whose size will be taken into account for proving
termination. The arguments of type $\vP$ are parameters and every rule
defining $\ff$ must be of the form $\ff\vp\vl\a r$ with $\vp\in\cX$,
$|\vp|=|\vP|$ and $|\vl|=|\vec\fB|$.

Moreover, the annotated type of a constructor $\fc:T_1\ldots T_n\A\fB$
is:
\[\tc^\cA=annot_\fB^\alpha(T_1)\A\ldots\A
annot_\fB^\alpha(T_n)\A\fB^{\fc^\cA(\alpha)}\]
\noindent
with $\fc^\cA(\alpha)=\infty$ if $\fc$ is non-recursive, and
$\fc^\cA(\alpha)=\fs(\alpha)$ otherwise, where $\fs$ is a monotonic
unary symbol interpreted as the ordinal successor and such that
$a\lts\fs(a)$ for each $a$.


\smallskip

{\bf Termination criterion.} We assume given a well-founded
quasi-ordering $\ge_\cF$ on $\cF$ and, for each function symbol
$\ff:^s\vT\A\vec\fB^{\valpha_\ff}\A\fB^{\ff^\cA(\valpha_\ff)}$ and set
$X\in\{\cA,\kA\}$, an ordered domain $(D_\ff^X,<_\ff^X)$ and a
function $\z_\ff^X:X^{|\valpha_\ff|}\a D_\ff^X$ compatible with $\eqf$
(\ie $|\valpha_\ff|=|\valpha_\fg|$, $D_\ff^X=D_\fg^X$,
${<_\ff^X}={<_\fg^X}$ and $\z_\ff^X=\z_\fg^X$ whenever $\ff\eqf\fg$)
and such that $>_\ff^\kA$ is well-founded and
$\z_\ff^\kA(\I\va\mu)<_\ff^\kA\z_\ff^\kA(\I\vb\mu)$ whenever
$\z_\ff^\cA(\va)<_\ff^\cA\z_\ff^\cA(\vb)$ and $\mu:\cX\a\kA$.

Usual domains are $\kA^n$ ordered lexicographically, or the multisets
on $\kA$ ordered with the multiset extension of $>_\kA$.

\begin{theorem}[\cite{blanqui04rta}]
\label{thm-sbt}
Let $\cR$ be a {\em constructor} system. The relation $\ab\cup\ar$
terminates if, for each defined
$\ff:^s\vP\A\vec\fB^\valpha\A\fB^{\ff^\cA(\valpha)}$ and rule
$\ff\vp\vl\a r\in\cR$, there is an environment $\G$ and a size
substitution $(_\valpha^\va)$ such that:
\begin{itemize}
\item
pattern condition: for each $\t$, if $\vp\t\in\I\vP$ and
$\vl\t\in\I{\vec\fB}$ then there is $\nu$ such that, for each
$(x,T)\in\G$, $x\t\in\I{T}^\nu$ and $\I\va\nu\le o_{\vec\fB}(\vl\t)$;
\item
argument decreasingness: $\G\thc^sr:\fB^{\ff^\cA(\va)}$ where $\thc$
is defined in Figure {}\ref{fig-thc};
\item
size annotations monotonicity:
$\pos(\valpha,\ff^\cA(\valpha))\sle\pos^+(\ff^\cA(\valpha))$.
\end{itemize}
\end{theorem}

The termination criterion introduced in {}\cite{blanqui04rta} is not
expressed exactly like this. The pattern condition is replaced by
syntactic conditions implying the pattern condition, but the
termination proof is explicitly based on the pattern condition. This
condition means that $\va$ is a valid representation of the size of
$\vl$, whatever the instantiation of the variables of $\vl$ is, and
thus that any recursive call with arguments of size smaller than $\va$
is admissible. The existence of such a valid syntactic representation
depends on $\vl$ and the size annotations of constructors. With the
chosen annotations, the condition is not satisfied by some patterns
(whose type admits elements of size bigger than $\w$, Appendix
\version{A}{\ref{sec-pat-cond}}). This suggests to use a more precise
annotation for constructors.

\begin{figure}
\vsp[-5mm]
\centering
\caption{Computability closure\label{fig-thc}}
\vsp[3mm]
$\cfrac{\fg\ltf\ff,~\psi:\cX\a\cA}{\G\thc^s\fg:\tg^\cA\psi}$
+ variable, abstraction, application and subtyping rules of Fig.
{}\ref{fig-th}\\[3mm]

$\cfrac{\fg\eqf\ff\quad
\fg:^s\vU\A\vec\fC^\vbeta\A\fC^{\fg^\cA(\vbeta)}\quad
\G\thc^s\vu:\vU\quad
\G\thc^s\vm:\vec\fB^\vb\quad
\z_\ff^\cA(\vb)<_\ff^\cA\z_\ff^\cA(\va)}
{\G\thc^s\fg\vu\vm:\fC^{\fg^\cA(\vb)}}$
\vsp[-3mm]
\end{figure}

The expressive power of the criterion depends on $\cA$. Taking the
size algebra $\cA$ reduced to the successor symbol $\fs$ (the
decidability of which is proved in {}\cite{blanqui05csl}) is
sufficient to handle every primitive recursive function. As an
example, consider the recursor $\rec_T:\fO\A T\A(\fO\A
T)\A((\fN\A\fO)\A(\fN\A T)\A T)\A T$ on the type $\fO$ of Brouwer's
ordinals whose constructors are $\mf{0}:\fO$,
$\fs:\fO^\alpha\A\fO^{\fs\alpha}$ and
$\lim:(\fN\A\fO^\alpha)\A\fO^{\fs\alpha}$, where $\fN$ is the type of
natural numbers whose constructors are $\mf{0}:\fN$ and
$\fs:\fN^\alpha\A\fN^{\fs\alpha}$:

\begin{rewc}
\rec\mf{0}uvw & u\\
\rec(\fs x)uvw & vx(\rec xuvw)\\
\rec(\lim f)uvw & wf(\l n\rec(fn)uvw)\\
\end{rewc}

For instance, with $f:\fN\A\fO^\alpha$, we have $\lim
f:\fO^{\fs\alpha}$, $fn:\fO^\alpha$ and $\fs\alpha>_\cA\alpha$.

An example of non-simply terminating system satisfying the criterion
is the following system defining a division function
$/:\fN^\alpha\A\fN\A\fN^\alpha$ by using a subtraction function
$-:\fN^\alpha\A\fN\A\fN^\alpha$.

\begin{center}
\begin{tabular}{c@{\hsp[1cm]}c}
\begin{rew}
-x\mf{0} & x\\
-\mf{0}x & \mf{0}\\
-(\fs x)(\fs y) & -xy\\
\end{rew}
&\begin{rew}
/\mf{0}x & \mf{0}\\
/(\fs x)y & \fs(/(-xy)y)\\
\end{rew}
\end{tabular}
\end{center}

Indeed, with $x:\fN^x$, we have $\fs x:\fN^{\fs x}$, $-xy:\fN^x$ and
$\fs x>_\cA x$.

\section{Annotating constructor types with a $\max$ symbol}
\label{sec-sup}

In this section, we simplify the previous termination criterion by
annotating constructor types in an algebra made of the following
symbols:

\begin{itemize}
\item
$\mf{0}\in\S_0$ interpreted as the ordinal $0$;
\item
$\fs\in\S_1$ interpreted as the successor ordinal;
\item
$\max\in\S_2$ interpreted as the max on ordinals.
\end{itemize}

For the annotated type of a constructor $\fc:T_1\ldots T_n\A\fB$, we
now take:
\[\tc^\cA=annot_\fB^{\alpha_1}(T_1)\A\ldots\A
annot_\fB^{\alpha_n}(T_n)\A\fB^{\fc_\cA(\alpha_1,\ldots,\alpha_n)}\]
with $\valpha$ distinct variables, $\fc_\cA(\valpha)=\mf{0}$ if $\fc$
is non-recursive, and $\fc_\cA(\valpha)=\fs(\max(\alpha_i\mid
i\in\ind(\fc)))$ otherwise, where $\max(\alpha_1,\ldots,\alpha_{k+1})
=\max(\alpha_1,\max(\alpha_2,\ldots,\alpha_{k+1}))$ and
$\max(\alpha_1)=\alpha_1$.

This does not affect the correctness of Theorem {}\ref{thm-sbt} since,
in this case too, one can prove that constructors are computable:
$\fc\in\I{\tc^\cA}^\mu$ for each $\mu$.

Moreover, now, both constructors and defined symbols have a type of
the form $annot_{\fB_1}^{\alpha_1}(T_1)\A\ldots\A
annot_{\fB_n}^{\alpha_n}(T_n)\A\fB^{\ff^\cA(\valpha)}$ with $\valpha$
distinct variables.

This means that a constructor can be applied to any sequence of
arguments without having to use subtyping. Indeed, previously, not all
constructor applications were possible (take $\fc xy$ with
$\fc:\fB^\alpha\A\fB^\alpha\A\fb^{\fs\alpha}$, $x:\fB^x$ and
$y:\fB^y$) and some constructor applications required subtyping (take
$\fc x(\fd x)$ with $\fc:\fB^\alpha\A\fB^\alpha\A\fb^{\fs\alpha}$,
$\fd:\fB^\alpha\A\fB^{\fs\alpha}$ and $x:\fB^x$).

We can therefore postpone subtyping after typing without losing much
expressive power . It follows that every term has a most general type
given by a simplified version of the type inference system $\th^i$ of
{}\cite{blanqui05csl} using unification only (see Appendix
\version{B}{\ref{sec-typ-inf}}).

Moreover, the pattern and monotonicity conditions can always be
satisfied by defining, for each symbol $\ff:^s\vP\A\vec\fB^\valpha\A
U$ and rule $\ff\vp\vl\a r\in\cR$, $\va$ as $\s(\vl)$ where $\s(x)=x$
and $\s(\fc\vt)=\fc^\cA(\s(\vt))$, and $\G$ as the set of pairs
$(x,T)$ such that $x\in\cX(\ff\vp\vl)$ and $T$ is:

\begin{itemize}
\item
$P_i$ if $x=p_i$,
\item
$\fB_i^x$ if $x=l_i$,
\item
$annot_{\fB_i}^x(T)$ if $\fc\vu x\vv$ is a subterm of $l_i$ and
$\fc:\vU\A T\A\vV\A\fC$.
\end{itemize}

Note that, if $\G\th t:T$ and $t$ is a non-variable pattern then there
is a base type $\fB$ such that $\G\th^it:\fB^{\s(t)}$. So, $\s(t)$ is
the most general size of $t$.

\comment{
\begin{lemma}
\label{lem-pat-cond}
For each $\t$, if $\vp\t\in\I\vP$ and $\vl\t\in\I{\vec\fB}$ then there
is $\nu$ such that, for each $(x,T)\in\G$, $x\t\in\I{T}^\nu$ and
$\va\nu=o_{\vec\fB}(\vl\t)$.
\end{lemma}

Hence, by annotating constructor types with $\max$, we get a simpler
termination criterion: }

\begin{theorem}
\label{thm-sbt-plus}
Let $\cR$ be a constructor system. The relation $\ab\cup\ar$
terminates if, for each
$\ff:^s\vP\A\vec\fB^\valpha\A\fB^{\ff^\cA(\valpha)}$ and rule
$\ff\vp\vl\a r\in\cR$, we have:
\begin{itemize}
\item
argument decreasingness: $\G\thc^i r:\fB^a$ and
$a\le_\o\cA\ff^\cA(\va)$ where $\G$ and $\va=\s(\vl)$ are defined just
before and $\thc^i$ is the type inference system $\th^i$
{}\cite{blanqui05csl} (see Appendix \version{B}{\ref{sec-typ-inf}})
with function applications restricted as in Figure {}\ref{fig-thc}.
\end{itemize}
\end{theorem}

\noindent
The proof is given in Appendix \version{C}{\ref{sec-sbt-plus}}. In the
following, we say that $\cR$ {\em SB-terminates} if $\cR$ satisfies
the conditions of Theorem {}\ref{thm-sbt-plus}.

\section{First-order semantic labelling}
\label{sec-fosl}

Semantic labelling is a transformation technique introduced by Hans
Zantema for proving the termination of first-order rewrite systems
{}\cite{zantema95fi}. It consists in labelling function symbols by
using some model of the rewrite system.

Let $\cF$ be a first-order signature and $\cM$ be an $\cF$-algebra
equipped with a partial order $\le_\cM$. For each $\ff\in\cF_n$, we
assume given a non-empty poset $(S^\ff,\le_\ff)$ and a {\em labelling
function} $\pi_\ff:M^n\a S^\ff$. Then, let $\o\cF$ be the signature
such that $\o\cF_n=\{\ff_a\mid \ff\in\cF_n,a\in S^\ff\}$.

The {\em labelling} of a term wrt a valuation $\mu:\cX\a M$ is defined
as follows: $lab^\mu(x)=x$ and $lab^\mu(\ff(t_1,\ldots,t_n))=
\ff_{\pi_\ff(\I{t_1}\mu,\ldots,\I{t_n}\mu)}(lab^\mu(t_1),\ldots,lab^\mu(t_n))$.

The fundamental theorem of semantic labelling is then:

\begin{theorem}[\cite{zantema95fi}]
\label{thm-fosl}
Given a rewrite system $\cR$, an ordered $\cF$-algebra $(\cM,\le_\cM)$
and a labelling system $(S^\ff,\le_\ff,\pi_\ff)_{\ff\in\cF}$, the
relation $\ar$ terminates if:

\begin{enumerate}
\item
$\cM$ is a quasi-model of $\cR$, that is:
\begin{itemize}
\item
for each rule $l\a r\in\cR$ and valuation $\mu:\cX\a M$,
$\I{l}\mu\ge_\cM\I{r}\mu$,
\item
for each $\ff\in\cF$, $\ff^\cM$ is monotonic;
\end{itemize}
\item
for each $\ff\in\cF$, $\pi_\ff$ is monotonic;
\item
the relation $\a_{lab(\cR)\cup Decr}$ terminates where:\\
$lab(\cR)=\{lab^\mu(l)\a lab^\mu(r)\mid l\a r\in\cR,\mu:\cX\a M\}$,\\
$Decr=\{\ff_a(x_1,\ldots,x_n)\a\ff_b(x_1,\ldots,x_n)\mid
\ff\in\cF,a>_\ff b\}$.
\end{enumerate}
\end{theorem}


For instance, by taking $M=\bN$, $\mf{0}^\cM=0$, $\fs^\cM(x)=x+1$,
$-^\cM(x,y)=x$ and $/^\cM(x,y)=x$, and by labelling $-$ and $/$ by the
semantics of their first argument, we get the following {\em infinite}
system which is easily proved terminating:

\begin{center}
\begin{tabular}{c@{\hsp[1cm]}c}
\begin{rew}
-_ix\mf{0} & x & ~(i\in\bN)\\
-_0\mf{0}x & \mf{0}\\
-_{i+1}(\fs x)(\fs y) & -_ixy & ~(i\in\bN)\\
\end{rew}
&\begin{rew}
/_0\mf{0}x & \mf{0}\\
/_{i+1}(\fs x)y & \fs(/_i(-_ixy)y) & ~(i\in\bN)\\
\end{rew}
\end{tabular}
\end{center}

\section{First-order case}
\label{sec-fo}

The reader may have already noticed some similarity between semantic
labelling and size annotations. We here render it more explicit by
giving a new proof of the correctness of SB-termination using semantic
labelling.

\comment{Let $\cF$ be a first-order signature. The set $\cT(\cF,\cX)$ of
first-order terms can be seen as a strict subset of the set of
simply-typed $\l$-terms over $\cF$ by taking $\cB=\{o\}$ and $f:o^n\A
o$ for every $f\in\cF_n$.}

In the first-order case, the interpretation of a base type does not
require transfinite iteration: all sizes are smaller than $\w$ and
$\kA=\bN$ {}\cite{blanqui05mscs}. Moreover, by taking $\G(x)=\fB^x$
for each $x$ of type $\fB$, every term $t$ has a most general size
$\s(t)$ given by its most general type: $\G\th^it:\fC^{\s(t)}$. This
function $\s$ extends to all terms the function $\s$ defined in the
previous section by taking
$\s(\ff(t_1,\ldots,t_n))=\ff^\cA(\s(t_1),\ldots,\s(t_n))$ for each
defined symbol $\ff$.


\begin{theorem}
\label{thm-fo}
SB-termination implies termination if:

\begin{itemize}
\item
$\cR$ is finitely branching and the set of constructors of each type
  $\fB$ is finite;
\item
for each defined symbol $\ff$, $\ff^\cA$ and $\z_\ff^\cA$ are monotonic.
\end{itemize}
\end{theorem}

\begin{proof}
For the interpretation domain, we take $M=\kA=\bN$ which has a
structure of poset with $\le_\cM=\le_\kA=\le_\bN$.

If $\ff^\cA$ is not the constant function equal to $\infty$
($\ff^\cA\neq\infty$ for short), which is the case of constructors,
then let $\ff^\cM(\vec\ka)=\I{\ff^\cA(\valpha)}\mu$ where
$\valpha\mu=\vec\ka$.

When $\ff^\cA=\infty$, we proceed in a way similar to {\em predictive
labelling} {}\cite{hirokawa06rta}, a variant of semantic labelling
where only the semantics of {\em usable symbols} need to be given when
$\cM$ is a {\em $\sqcup$-algebra} (all {\em finite} subsets of $M$
have a lub wrt $\le_\cM$), which is the case of $\bN$. Here, the
notions of usable symbols and rules are not necessary and a semantics
can be given to all symbols thanks to the strong assumptions of
SB-termination.

Let \comment{$>^\kA$ be the lexicographic combination of $\gtf$ and $>_\ff^\kA$
as follows:} $(\ff,\vx)>^\kA(\fg,\vy)$ if $\ff\gtf\fg$ or $\ff\eqf\fg$
and $\z_\ff^\kA(\vx)>_\ff^\kA\z_\ff^\kA(\vy)$. The relation $>^\kA$ is
well-founded since the relations $\gtf$ and $>_\ff^\kA$ are
well-founded. We then define $\ff^\cM$ by induction on $>^\kA$ by
taking $\ff^\cM(\vec\ka)=max(\{0\}\cup\{\I{r}\mu\mid\ff\vl\a
r\in\cR,\mu:\cX\a\kA,\I\vl\mu\le\vec\ka\})$. This function is well
defined since:

\begin{itemize}
\item
For each subterm $\fg\vm$ in $r$,
$(\ff,\s(\vl))>^\cA(\fg,\s(\vm))$. Assume that $\ff\eqf\fg$. Then,
$\s(\vl)\gts\s(\vm)$. Hence, for each symbol $\ff$ occuring in $\vl$
or $\vm$, $\ff^\cA\neq\infty$. Therefore, $\I\vl\mu=\I{\s(\vl)}\mu$,
$\I\vm\mu=\I{\s(\vm)}\mu$ and $(\ff,\I\vl\mu)>^\kA(\fg,\I\vm\mu)$.
\item
The set $\{(\ff\vl\a r,\mu)\mid\ff\vl\a r\in\cR,\I\vl\mu\le\vec\ka\}$
is finite. Indeed, since $\vl$ are patterns and constructors are
interpreted by monotonic and strictly extensive functions (\ie
$\fc^\cA(\valpha)\ges\fs(\max(\alpha_i\mid i\in\ind(\fc)))$),
$\I\vl\mu$ is strictly monotonic wrt $\mu$ and the height of $\vl$. We
cannot have an infinite set of $\vl$'s of bounded height since, for
each base type $\fB$, the set of constructors of type $\fB$ is
finite. And we cannot have an infinite set of $r$'s since $\cR$ is
finitely branching.
\end{itemize}

We do not label the constructors, \ie we take any singleton set for
$S^\fc$ and the unique (constant) function from $M^n$ to $S^\fc$ for
$\pi_\fc$. For any other symbol $\ff$, we take $S^\ff=D_\ff^\kA$ which
is well-founded wrt $>_\ff$, and $\pi_\ff=\z_\ff^\kA$.


\begin{enumerate}
\item
$\cM$ is a quasi-model of $\cR$:
\begin{itemize}
\item
Let $\ff:^s\vP\A\vec\fB^\valpha\A\fB^{\ff^\cA(\valpha)}$, $l\a
r\in\cR$ with $l=\ff\vp\vl$, and $\mu:\cX\a M$. We have
$\I{l}\mu=\ff^\cM(\vec\ka)$ where $\vec\ka=\I\vl\mu$. If
$\ff^\cA=\infty$, then
$\ff^\cM(\vec\ka)=max(\{0\}\cup\{\I{r}\mu\mid\ff\vl\a
r\in\cR,\mu:\cX\a\kA,\I\vl\mu\le\vec\ka\})$ and $\I{l}\mu\ge\I{r}\mu$.
Assume now that $\ff^\cA\neq\infty$. Since $\G\thc r:^i\fB^a$ and
$a\le_\o\cA\ff^\cA(\va)$, we have $\s(r)=a\le_\o\cA\ff^\cA(\va)=\s(l)$
where $\va=\s(\vl)$. By definition of $\G$ and $\s$, for each $i$,
$a_i\neq\infty$ ($\va\neq\infty$ for short). Therefore,
$\s(l)\neq\infty$ and $\s(r)\les\s(l)$. Hence,
$\I{l}\mu=\s(l)\mu\le_\kA\s(r)\mu=\I{r}\mu$ since $\le_\kA$ is a model
of $\les$.

\item
If $\ff$ is a non-recursive constructor, then $\ff^\cM(\vec\ka)=0$ is
monotonic. If $\ff$ is a recursive constructor, then
$\ff^\cM(\vec\ka)=sup\{\ka_i\mid i\in\ind(\fc)\}+1$ is monotonic. If
$\ff^\cA\neq\infty$, then $\ff^\cM(\vec\ka)=\I{\ff^\cA(\valpha)}\mu$
where $\valpha\mu=\vec\ka$ is monotonic since $\ff^\cA$ is monotonic
by assumption. Finally, if $\ff^\cA=\infty$, then
$\ff^\cM(\vec\ka)=max(\{0\}\cup\{\I{r}\mu\mid\ff\vl\a
r\in\cR,\mu:\cX\a\kA,\I\vl\mu\le\vec\ka\})$ is monotonic.
\end{itemize}

\item
If $\ff$ is a defined symbol, then the function $\pi_\ff$ is monotonic
by assumption. If $\ff$ is a constructor, then the constant function
$\pi_\ff$ is monotonic too.

\item
We now prove that $\a_{lab(\cR)\cup Decr}$ is precedence-terminating
(PT), \ie there is a well-founded relation $>$ on symbols such that,
for each rule $\ff\vl\a r\in lab(\cR)\cup Decr$, every symbol
occurring in $r$ is strictly smaller than $\ff$
{}\cite{middeldorp96cade}.

Let \comment{$<$ be the lexicographic combination of $\ltf$ and $<_\ff^\kA$ as
follows:} $\fg_a<\ff_b$ if $\fg\ltf\ff$ or $\fg\eqf\ff$ and $a<_\ff^\kA
b$. The relation $>$ is well-founded since both $\gtf$ and $>_\ff^\kA$
are well-founded.

$Decr$ is clearly PT wrt $>$. Let now $\ff\vl\a r\in\cR$, $\mu:\cX\a
M$ and $\fg\vt$ be a subterm of $r$. The label of $\ff$ is
$a=\pi_\ff(\I{\vl}\mu)=\z_\ff^\kA(\I{\s(\vl)}\mu)$ and the label of
$\fg$ is $b=\z_\ff^\kA(\I{\s(\vm)}\mu)$. By assumption,
$(\ff,\vl)>^\cA(\fg,\vm)$. Therefore, $a>_\ff^\kA b$.\qed
\end{enumerate}
\end{proof}

It is interesting to note that we could also have taken $M=\cA$,
assuming that $<_\ff^\cA$ is stable by substitution
($\z_\ff^\cA(\va\t)<_\ff^\cA\z_\ff^\cA(\vb\t)$ whenever
$\z_\ff^\cA(\va)<_\ff^\cA\z_\ff^\cA(\vb)$). The system labelled with
$\cA$ is a syntactic approximation of the system labelled with
$\kA$. Although less powerful {\em a priori}, it may be interesting
since it provides a {\em finite} representation of the infinite
$\kA$-labelled system.

\comment{
For instance, by taking $\mf{0}^\cM=0$, $\fs^\cM(x)=\fs x$,
$-^\cM(x,y)=x$ and $/^\cM(x,y)=x$, and by labelling $-$ and $/$ by the
semantics of their first argument, we get that every $\kA$-labelled
rule is an $i$-instance of one of the following $\cA$-labelled rules:

\begin{center}
\begin{tabular}{c@{\hsp[1cm]}c}
\begin{rew}
-_ix\mf{0} & x & ~(i\in\cX)\\
-_0\mf{0}x & \mf{0}\\
-_{\fs i}(\fs x)(\fs y) & -_ixy & ~(i\in\cX)\\
\end{rew}
&\begin{rew}
/_0\mf{0}x & \mf{0}\\
/_{\fs i}(\fs x)y & \fs(/_i(-_ixy)y) & ~(i\in\cX)\\
\end{rew}
\end{tabular}
\end{center}
}

Finally, we see from the proof that the system does not need to be
constructor:\comment{ if the following conditions are satisfied (an example is
given in Appendix \version{D}{\ref{sec-non-constr}}):}

\begin{theorem}
\label{thm-sbt-plus-plus}
Theorem {}\ref{thm-fo} holds for any (non-constructor) system $\cR$
such that, for each rule $\ff\vl\a r\in\cR$ with $\ff^\cA=\infty$ and
subterm $\fg\vm$ in $\vl$:

\begin{itemize}
\item
$\fg^\cA$ is monotonic and strictly extensive:
$\fg^\cA(\valpha)\ges\fs(\max(\alpha_i\mid i\in\ind(\fc)))$,
\item
if $\fg^\cA=\infty$, then $\fg\ltf\ff$ or $\fg\eqf\ff$ and
$\z_\ff^\cA(\s(\vm))<_\ff^\cA\z_\ff^\cA(\s(\vl))$.
\end{itemize}
\end{theorem}

{\bf Example:} assuming that $\fA$ is the $\A$-type constructor, then
the expression $\fF nuv$ represents the set of $n$-ary functions from
$u$ to $v$.

\vsp[-3mm]
\begin{center}
\begin{tabular}{cc}
\begin{rew}
+\mf{0}y & y\\
+(\fs x)y & \fs(+xy)\\
+(+xy)z & +x(+yz)\\
\end{rew}
&\quad
\begin{rew}
\fF\mf{0}uv & v\\
\fF(\fs x)uv & \fA u(\fF xuv)\\
\fF(+xy)uv & \fF xu(\fF yuv)\\
\end{rew}
\end{tabular}
\end{center}
\vsp[-3mm]

Take $+^\cA(x,y)=\z_+(x,y)=a=2x+y+1$, $\fF^\cA=\infty$ and
$\z_\fF(x,u,v)=x$. The interpretation of $\fF^\cM$ is well-defined
since $x<a$ and $y<a$. The labelled system that we obtain (where
$b=2y+z+1$) is precedence-terminating:

\vsp[-3mm]
\begin{center}
\begin{tabular}{cc}
\begin{rew}
+_{y+1}\mf{0}y & y\\
+_{a+2}(\fs x)y & \fs(+_axy)\\
+_{2a+z+1}(+_axy)z & +_{2x+b+1}x(+_byz)\\
\end{rew}
&
\begin{rew}
\fF_{0}\mf{0}uv & v\\
\fF_{x+1}(\fs x)uv & \fA u(\fF_xxuv)\\
\fF_a(+_axy)uv & \fF_xxu(\fF_yyuv)\\
\end{rew}
\end{tabular}
\end{center}
\vsp[-6mm]

\section{Higher-order semantic labelling}
\label{sec-hosl}

Semantic labelling was extended by Hamana {}\cite{hamana07ppdp} to
second-order Inductive Data Type Systems (IDTSs) with higher-order
pattern-matching {}\cite{blanqui00rta}. IDTSs are a typed version of
Klop's Combinatory Reduction Systems (CRSs) {}\cite{klop93tcs} whose
categorical semantics based on binding algebras and $\cF$-monoids
{}\cite{fiore99lics} is studied by the same author and proved complete
for termination {}\cite{hamana05rta}.

The fundamental theorem of higher-order semantic labelling can be
stated exactly as in the first-order case, but the notion of model is
more involved.



\smallskip

{\bf CRSs and IDTSs.}\comment{In CRSs, a term is either a variable
$x\in\cX$, an abstraction $\l xt$ or the application $f(t_1,\ldots,
t_n)$ of a function symbol $f\in\cF$ of arity $n$ to $n$ terms
$t_1,\ldots,t_n$.} In CRSs, function symbols have a fixed arity. {\em
Meta-terms} extend terms with the application $Z(t_1,\ldots,t_n)$ of
a meta-variable $Z\in\cZ$ of arity $n$ to $n$ meta-terms
$t_1,\ldots,t_n$.

An assignment $\t$ maps every meta-variable of arity $n$ to a term of
the form $\l x_1..\l x_nt$. Its application to a meta-term $t$,
written $t\t$, is defined as follows:

\vsp[-3mm]
\begin{itemize}
\item
$x\t=x$, $(\l xt)\t=\l x(t\t)$ and
$\ff(t_1,\ldots,t_n)\t=\ff(t_1\t,\ldots,t_n\t)$;
\item
for $\t(Z)=\l x_1..\l x_nt$,
$Z(t_1,\ldots,t_n)\t=t\{x_1\to t_1\t,\ldots,x_n\to t_n\t\}$.
\end{itemize}
\vsp[-3mm]

A rule is a pair of {\em meta-terms} $l\a r$ such that $l$ is a
higher-order pattern {}\cite{miller89elp}.\comment{ The relation generated by a
set $\cR$ of rules, written $\ar$, is defined as follows:

\begin{itemize}
\item
$l\t\ar t\t$ for each $l\a r\in\cR$ and $\t:\cZ\a\cT$,
\item
$\l xt\ar \l xu$ and
$\ff(\ldots,t,\ldots)\ar\ff(\ldots,u,\ldots)$ if $t\ar u$.
\end{itemize}

For instance, $\b$-reduction is generated by the rule
$@(\l(\l xZ(x)),X)\a Z(X)$, and $\eta$-reduction by the {\em
non-conditional} rule $\l(\l x@(Z,x))\a Z$. Note that, for $\b$, $Z$ is
a meta-variable of arity $1$ while, for $\eta$, $Z$ is of arity $0$,
so that $Z$ cannot be substituted by a term with $x$ free in it.
}

In IDTSs, variables, meta-variables and symbols are equipped with
types over a discrete category $\bB$ of base types. However, Hamana
only considers {\em structural meta-terms} where abstractions only
appear as arguments of a function symbol, variables are restricted to
base types, meta-variables to first-order types and function symbols
to second-order types. But, as already noticed by Hamana, this is
sufficient to handle any rewrite system (see Section
{}\ref{sec-ho}). Let $I^\cZ_B(\G)$ be the set of structural meta-terms
of type $B$ in $\G$ whose meta-variables are in $\cZ$.

\comment{Formally, the set $I^\cZ_B(\G)$ of structural meta-terms of type $B$
in $\G$ whose meta-variables are in $\cZ$, written $\G\th_I^\cZ t:B$,
is defined as follows:

\begin{itemize}
\item
$\G\th_I^\cZ x:B$ if $(x,B)\in\G$,
\item
$\G\th_I^\cZ Z(\vt):B$ if $Z:\vB\A B\in\cZ$ and, for each $i$,
$\G\th_I^\cZ t_i:B_i$.
\item
$\G\th_I^\cZ \ff(\l\vx_1t_1,\ldots,\l\vx_nt_n):B$ if $\ff:(\vB_1\A
B_1)\A\ldots\A(\vB_n\A B_n)\A B$ and, for each $i$,
$\G,\vx:\vB_i\th_I^\cZ t_i:B_i$ for pairwise distinct fresh variables
$\vx$.
\end{itemize}
}


\smallskip

{\bf Models.} The key idea of binding algebras {}\cite{fiore99lics} is
to interpret variables by natural numbers using De Bruijn levels
\comment{{}\cite{debruijn72}}, and to handle bound variables by extending the
interpretation to typing environments.

Let $\bF$ be the category whose objects are the finite cardinals and
whose arrows from $n$ to $p$ are all the functions from $n$ to
$p$. Let $\bE$ be the (slice) category of typing environments whose
objects are the maps $\G:n\a\bB$ and whose arrows from $\G:n\a\bB$ to
$\D:p\a\bB$ are the functions $\r:n\a p$ such that $\G={\D\circ\r}$.

Given $\G:n\a\bB$, let $\G+B:n+1\a\bB$ be the environment such that
$(\G+B)(n)=B$ and $(\G+B)(k)=\G(k)$ if $k<n$.

Let $\bM$ be the functor category ${(\Set^\bE)}^\bB$. An object of
$\bM$ (presheaf) is given by a family of sets $M_B(\G)$ for every base
type $B$ and environment $\G$ and, for every base type $B$ and arrow
$f:\G\a\D$, a function $M_B(f):M_B(\G)\a M_B(\D)$ such that
$M_B(id_\G)=id_{M_B(\G)}$ and $M_B(f\circ g)=M_B(f)\circ M_B(g)$. An
arrow $\alpha:M\a N$ in $\bM$ is a natural transformation, \ie a
family of functions $\alpha_B(\G):M_B(\G)\a N_B(\G)$ such that, for
each $\r:\G\a\D$, $\alpha_B(\D)\circ M_B(\r)=N_B(\r)\circ\alpha_B(\G)$.

Given $M\in\bM$, $\G\in\bE$ and $\vB\in\bB$, let $up_\G^\vB(M):M(\G)\a
M(\G+\vB)$ be the arrow equal to $M(id_\G+0_\D)$ where $0_\D$ is the
unique morphism from $0$ to $\D$.


An {\em $\cX+\cF$-algebra} $\cM$ is given by a presheaf $M\in\bM$, an
interpretation of variables $\io:\cX\a\cM$ and, for every symbol
$\ff:(\vB_1\A B_1)\A\ldots\A(\vB_n\A B_n)\A B$ and environment $\G$,
an arrow $\ff^\cM(\G):\prod_{i=1}^n M_{B_i}(\G+\vB_i)\a M_B(\G)$.

The category $\bM$ forms a monoidal category with unit $\cX$ and
product $\bullet$ such that $(M\bullet N)_B(\G)$ is the set of
equivalence classes on the set of pairs $(t,\vu)$ with $t\in M_B(\D)$
and $u_i\in N_{\D(i)}(\G)$ for some $\D$, modulo the equivalence
relation $\sim$ such that $(t,\vu)\sim(t',\vu')$ if there is
$\r:\D\a\D'$ for which $t\in M_B(\D)$, $t'=M_B(\r)(t)$ and
$u'_{\r(i)}=u_i$.

To interpret substitutions, $M$ must be an {\em $\cF$-monoid}, \ie a
monoid $(M,\mu:M^2\a M)$ compatible with the structure of
$\cF$-algebra {}\cite{hamana07ppdp} (see Appendix
\version{E}{\ref{sec-monoid}}).

\comment{
\begin{itemize}
\item
$\mu_B(\G)(\io_B(\D)(i),\vu)=u_i$;
\item
$\mu_B(\G)(t,\io_{\D(1)}(\G)(1)\ldots\io_{\D(p)}(\G)(p))=t$;
\item
for $t\in M_B(\T)$, $u_i\in M_{\T(i)}(\D)$ and $v_i\in M_{\D(i)}(\G)$,\\
$\mu_B(\G)(\mu_B(\D)(t,\vu),\vv)=\mu_B(\G)(t,\mu_{\T(1)}(\G)(u_1,\vv)\ldots\mu_{\T(p)}(\G)(u_p,\vv))$;
\item
for $\ff:(\vB_1\A B_1)\A\ldots\A(\vB_n\A B_n)\A B$ and
$\G_i=\G+\vB_i$,\\ $\mu_B(\G)(\ff^\cM(\D)(\vt),\vu)
=\ff^\cM(\G)(\mu_{B_1}(\G_1)(t_1,\vv_1),\ldots,\mu_{B_n}(\G_n)(t_n,\vv_n))$\\
where $v_{i,j}=up_\G^{\vB_i}(u_j)$ if $j<|\D|$, and
$v_{i,j}=|\G|+j-|\D|$ otherwise.
\end{itemize}
}

The presheaf $I^\vide$ equipped with the product
$\mu_B(\G)(t,\vu)=t\{i\to u_i\}$ (simultaneous substitution) is
initial in the category of $\cF$-monoids {}\cite{hamana05rta}. Hence,
for each $\cF$-monoid $\cM$, there is a unique morphism
$!^\cM:I^\vide\a M$.

\comment{
In the category of $\cF$-monoids, the presheaf of meta-terms $I^\cZ$
equipped with the product $\mu_B(\G)(t,\vu)=t\{x_1\to
u_1,\ldots,x_n\to u_n\}$ (simultaneous substitution) is free. Hence,
given an $\cF$-monoid $M$, any valuation $\phi:\cZ\a M$ can be
uniquely extended into an $\cF$-monoid morphism $\phi^*:I^\cZ\a M$
such that:

\begin{itemize}
\item
$\phi^*_B(\G)(x)=\io_B(\G)(x)$;
\item
for $Z:\vB\A B$,\\
$\phi^*_B(\G)(Z(t_1,\ldots,t_n))
=\mu_B(\G)(\phi_B(\vB)(Z),\phi^*_{B_i}(\G)(t_1)\ldots\phi^*_{B_n}(\G)(t_n))$;
\item
for $\ff:(\vB_1\A B_1)\A\ldots\A(\vB_n\A B_n)\A B$ and
$\G_i=\G,\vx_i:\vB_i$,\\
$\phi^*_B(\G)(\ff(\l\vx_1t_1,\ldots,\l\vx_nt_n))
=(\ff^\cM)_B(\G)(\phi^*_{B_1}(\G_1)(t_1),\ldots,\phi^*_{B_n}(\G_n)(t_n))$.
\end{itemize}

In the following, we assume given an $\cF$-monoid $\cM$ and let
$!^\cM=\vide^*:I^\vide\a M$.
}


\smallskip

{\bf Labelling.} As in the first-order case, for each $\ff:(\vB_1\A
B_1)\A\ldots\A(\vB_n\A B_n)\A B$, we assume given a non-empty poset
$(S^\ff,\le_\ff)$ for {\em labels} and a {\em labelling function}
$\pi_\ff(\G):\prod_{i=1}^nM_{B_i}(\G+\vB_i)\a S^\ff$. Let
$\o\cF_n=\{\ff_a\mid\ff\in\cF_n,a\in S^\ff\}$. Note that the set of
labelled meta-terms has a structure of $\cF$-monoid
{}\cite{hamana07ppdp}.

The {\em labelling} of a meta-term wrt a valuation $\t:\cZ\a I^\vide$
is defined as follows:

\begin{itemize}
\item
$lab^\t_B(\G)(x)=x$;
\item
$lab^\t_B(\G)(Z(t_1,\ldots,t_n))=
Z(lab^\t_B(\G)(t_1),\ldots,lab^\t_B(\G)(t_n))$;
\item
for $\ff:(\vB_1\A B_1)\A\ldots\A(\vB_n\A B_n)\A B$ and
$\G_i=\G,\vx_i:\vB_i$,\\
$lab^\t_B(\G)(\ff(\l\vx_1t_1,\ldots,\l\vx_nt_n))
=\ff_a(lab^\t_{B_1}(\G_1)(t_1),\ldots,lab^\t_{B_n}(\G_n)(t_n))$\\
where
$a=\pi_\ff(!^\cM_{B_1}(\G_1)(t_1\t),\ldots,!^\cM_{B_n}(\G_n)(t_n\t))$.
\end{itemize}

\comment{Given a labelled term $t$, let $|t|$ be the term obtained after
removing all labels. Note that the presheaf of labelled meta-terms
$\o{I}^\cZ$ has a structure of $\cF$-monoid for each valuation
$\t:\cZ\a I^\vide$ by taking:

\begin{itemize}
\item
$\mu^\t_B(\G)(i,\vu)=u_i$;
\item
for $Z:\vB\A B$,
$\mu^\t_B(\G)(Z(t_1,\ldots,t_n),\vu)=
Z(\mu^\t_{B_1}(\G)(t_1),\ldots,\mu^\t_{B_n}(\G)(t_n))$;
\item
for $\ff:(\vB_1\A B_1)\A\ldots\A(\vB_n\A B_n)\A B$, $\G_i=\G+\vB_i$
and $u_i\in\o{I}^{\t,\cZ}_{\D(i)}(\G)$,\\
$\mu^\t_B(\G)(\ff_a(\l\vx_1t_1,\ldots,\l\vx_nt_n),\vu)
=\ff_b(\mu^\t_{B_1}(\G_1)(t_1,\vv_1),\ldots,\mu^\t_{B_n}(\G_n)(t_n,\vv_n))$\\
where $b=\pi^\ff_B(\G)(!^\cM_{B_1}(\G_1)(|t_1|\t),\ldots,
!^\cM_{B_n}(\G_n)(|t_n|\t))$,\\ $v_{i,j}=up_\G^{\vB_i}(u_j)$ if
$j<|\D|$, and $v_{i,j}=|\G|+j-|\D|$ otherwise.
\end{itemize}
}

We can now state Hamana's theorem for higher-order semantic labelling.

\begin{theorem}[\cite{hamana07ppdp}]
\label{th-hosl}
Given a structural IDTS $\cR$, an ordered $\cF$-algebra
$(\cM,\le_\cM)$ and a labelling system
$(S^\ff,\le_\ff,\pi_\ff)_{\ff\in\cF}$, the relation $\ar$ terminates
if:

\begin{enumerate}
\item
$(\cM,\le_\cM)$ is a quasi-model of $\cR$, that is:
\begin{itemize}
\item
for each $l\a r:T\in\cR$, $\t:\cZ\a I^\vide$ and $\G$,
$!^\cM_B(\G)(l\t)\ge_{M_B(\G)}!^\cM_B(\G)(r\t)$,
\item
for each $\ff\in\cF$, $\ff^\cM$ is monotonic;
\end{itemize}

\item
for each $\ff\in\cF$, $\pi_\ff$ is monotonic;

\item
the relation $\a_{lab(\cR)\cup Decr}$ terminates, where:\\
$lab(\cR)=\{lab_B^\vide(\G)(l\t)\a lab_B^\vide(\G)(r\t)\mid l\a r:B\in\cR,
\t:\cZ\a I^\vide, \G\in\bE\}$,\\
$Decr=\{\ff_a(\ldots,\l\vx_iZ_i(\vx_i),\ldots)\a
\ff_b(\ldots,\l\vx_iZ_i(\vx_i),\ldots)\mid\ff\in\cF,{a>_\ff b}\}$.
\end{enumerate}
\end{theorem}

\section{Higher-order case}
\label{sec-ho}

In order to apply Hamana's higher-order semantic labelling, we first
need to translate into a structural IDTS not only the rewrite system
$\cR$ but also $\b$ itself.


\smallskip

{\bf Translation to structural IDTS.} Following Example 4.1 in
{}\cite{hamana07ppdp}, the relations $\b$ and $\cR$ can be encoded in
a structural IDTS as follows.

Let the set of {\em IDTS base types} $\bB$ be the set $\cT(\S)$ where
$\S_0=\cB$ is the set of base types, $\S_2=\{Arr\}$ and $\S_n=\vide$
otherwise. A simple type $T$ can then be translated into an IDTS base
type $\ps{T}$ by taking $\ps{T\A U}=Arr(\ps{T},\ps{U})$ and $\ps{T}=T$
if $T\in\cB$. Then, an environment $\G$ can be translated into an IDTS
environment $\ps\G$ by taking $\ps\vide=\vide$ and
$\ps{x:T,\G}=x:\ps{T},\ps\G$. Conversely, let $|T|$ be the simple type
such that $\ps{|T|}=T$.

Let the set of {\em IDTS function symbols} be the set $\ps\cF$ made of
the symbols $\ps\ff:\ps{T_1}\A\ldots\A\ps{T_n}\A\fB$ such that
$\ff:T_1\A\ldots\A T_n\A\fB$, and all the symbols $\l_T^U:(T\A U)\A
Arr(T,U)$ and $@_T^U:Arr(T,U)\A T\A U$ such that $T$ and $U$ are IDTS
base types. Note that only $\l_T^U$ has a second order type.

A simply-typed $\l$-term $t$ such that $\G\th t:T$ can then be
translated into an IDTS term $\ps{t}_\G$ such that
$\ps\G\th\ps{t}_\G:\ps{T}$ as follows:

\begin{itemize}
\item
$\ps{x}_\G=x$,
\item
$\ps{\l x^Tu}_\G=\l_\ps{T}^\ps{U}(\l x\ps{u}_{\G,x:T})$ if $\G,x:T\th
u:U$,
\item
for $\ff:T_1\A\ldots\A T_n\A\fB$ and $U_i=T_{i+1}\A\ldots\A T_n\A\fB$,\\
$\ps{\ff t_1\ldots t_k}_\G=\l_\ps{T_{k+1}}^\ps{U_{k+1}}(\l x_{k+1}\ldots\l_\ps{T_n}^\ps{U_n}(\l x_n\ps\ff(\ps{t_1}_\G,\ldots,\ps{t_k}_\G,x_{k+1},\ldots,x_n))...)$,
\item
$\ps{tu}_\G=@_\ps{U}^\ps{V}(\ps{t}_\G,\ps{u}_\G)$ if $\G\th t:U\A V$.
\end{itemize}

A rewrite rule $l\a r\in\cR$ is then translated into the IDTS rule
$\ps{l}\a\ps{r}$ where the free variables of $l$ are seen as nullary
meta-variables, and $\b$-rewriting is translated into the family of
IDTS rules $\ps\b=\bigcup_{T,U\in\bB}\b_T^U$ where $\b_T^U$ is:

\[@_T^U(\l_T^U(\l xZ(x)),X)\a Z(X)\]

\noindent
where $Z$ (resp. $X$) is a meta-variable of type $T\A U$
(resp. $T$). Note that only $\ps\b$ uses non-nullary meta-variables.

Then, $\ar\cup\ab$ terminates iff $\a_{\ps\cR\cup\ps\b}$ terminates
(Appendix \version{F}{\ref{sec-idts}}).


\smallskip

{\bf Interpretation domain.} We now define the interpretation domain
$M$ for interpreting $\ps\b\cup\ps\cR$. First, we interpret
environments as arrow types:

\begin{itemize}
\item
$M_T(\G)=N_{Arr(\G,T)}$ where:\\
$Arr(\vide,T)=T$ and
$Arr(\G+U,T)=Arr(\G,Arr(U,T))$.
\end{itemize}

As explained at the beginning of Section {}\ref{sec-sbt}, to every
base type $\fB\in\cB$ corresponds a limit ordinal $\w_\fB<\kA$ that is
the number of transfinite iterations of the monotonic function $F_\fB$
that is necessary to build the interpretation of $\fB$.

So, a first idea is to take $N_\fB=\w_\fB$ and the set of functions
from $N_T$ to $N_U$ for $N_{Arr(T,U)}$. But taking all functions
creates some problems. Consider for instance the constructor
$\lim:(\fN\A\fO)\A\fO$. We expect $\lim^\cM(\vide)(f)=sup\{f(n)\mid
n\in N_\fN\}+1$ to be a valid interpretation, but $sup\{f(n)\mid n\in
N_\fN\}+1$ is not in $N_\fO$ for each function $f$. We therefore need
to restrict $N_{Arr(T,U)}$ to the functions that correspond to (are
realized by) some $\l$-term.

Hence, let $N_T=\{x\mid\ex t\in\cT,~t\th_T x\}$ where $\th_T$ is
defined as follows:

\begin{itemize}
\item
$t\th_\fB\ka\in\w_\fB$ if $t\in\I\fB$ and $o_\fB(t)\ge\ka$,
\item
$v\th_{Arr(T,U)}f:N_T\a N_U$ if $v\in\I{|T|\A|U|}$ and $vt\th_U f(x)$
whenever $t\th_Tx$.
\end{itemize}

Then, we can now check that $sup\{f(n)\mid n\in N_\fN\}+1\in
N_\fO$. Indeed, if there are $v$ and $t$ such that
$v\th_{Arr(\fN,\fO)}f$ and $t\th_\fN n$, then $vt\th_\fO f(n)$ and
$\lim(v)\th_\fO sup\{f(n)\mid n\in N_\fN\}+1\in N_\fO$.

The action of $M$ on $\bE$-morphisms is defined as follows. Given
$f:\G\a\D$ with $\G:n\a\bB$ and $\D:p\a\bB$, let $M_T(f):M_T(\G)\a
M_T(\D)$ be the function mapping $x_0\in N_{Arr(\G,T)}$, $x_1\in
N_{\D(1)}$,
\ldots, $x_p\in N_{\D(p)}$ to $x_0(x_{f(1)},\ldots,x_{f(n)})$.
\comment{($x_{f(i)}:\D(f(i))=\G(i)$ by definition of $\bE$-morphisms).}

Finally, the sets $M_B(\G)$ and $N_T$ are ordered as follows:

\begin{itemize}
\item
$x\le_{M_B(\G)}y$ if $x\le_{N_{Arr(\G,B)}}y$ where:
\begin{itemize}
\item
$x\le_{N_\fB}y$ if $x\le y$,
\item
$f\le_{N_{Arr(T,U)}}g$ if $f(x)\le_{N_U}g(x)$ for each $x\in N_T$.
\end{itemize}
\end{itemize}


{\bf Interpretation of variables and function symbols.} As one can
expect, variables are interpreted by projections:
$\io_{\G(i)}(\G)(i)(\vx)=x_i$, $\l_T^U$ by the identity:
$(\l_T^U)^\cM(\G)(f)=f$, and $@_T^U$ by the application:
$(@_T^U)^\cM(\G)(f,x)(\vy)=f(\vy,x(\vy))$.

\comment{
\begin{itemize}
\item
$\io_{\G(i)}(\G)(i)(\vx)=x_i$,
\item
$(\l_T^U)^\cM(\G)(f)=f$,
\item
$(@_T^U)^\cM(\G)(f,x)(\vy)=f(\vy,x(\vy))$.
\end{itemize}
}

One can check that these functions are valid interpretations indeed,
\ie $\io_{\G(i)}(\G)(i)(\vx)\in N_{\G(i)}$ and
$(@_T^U)^\cM(\G)(f,x)(\vy)\in N_U$.

Moreover, we have $(@_T^U)^\cM(\G)(f,x)(\vx)=\mu_U(\G)(f,\vp x)$ where
$p_i=\io_{\G(i)}(\G)(i)$ and $\mu$ is the monoidal product
$\mu_B(\G)(t,u_1\ldots u_n)(\vx)=t(u_1(\vx),\ldots,u_n(\vx))$.

We can then verify that $\ps\b$ is valid if $(M,\mu)$ is an
$\cF$-monoid, and that $(M,\mu)$ is an $\cF$-monoid if, for each $\ff$
and $\G$,
$\ff^\cM(\G)(\vx)(\vy)=\ff^\cM(\vide)(x_1(\vy),\ldots,x_n(\vy))$
(Appendix \version{G}{\ref{sec-beta}}).

\comment{
\begin{lemma}
\label{lem-beta-valid}
If $(M,\mu)$ is an $\cF$-monoid, then $\ps\b$ is valid in $M$.
\end{lemma}

And $M$ is an $\cF$-monoid if the functions $\ff^\cM$ are {\em
parametric} wrt environments:

\begin{lemma}
\label{lem-monoid}
$(M,\mu)$ is an $\cF$-monoid if
$\ff^\cM(\G)(\vx)(\vy)=\ff^\cM(\vide)(x_1(\vy),\ldots,x_n(\vy))$.
\end{lemma}
}

One can see that $(\l_T^U)^\cM$ and $(@_T^U)^\cM$ satisfy this
property. Moreover, for each term $t\in I^\vide_T(\G)$, we have
$!_T^\cM(x_1:T_1...\,x_n:T_n)(t)(\va)=\I{t}\mu$ where $x_i\mu=a_i$ and:

\begin{center}
$\I{x}\mu=\mu(x)$\quad
$\I{@_T^U(v,t)}\mu=\I{v}\mu(\I{t}\mu)$\quad
$\I{\l_T^U(\l xu)}\mu=a\to\I{u}\mu_x^a$\quad
$\I{\ff(\vt)}\mu=\ff^\cM(\vide)(\I\vt\mu)$\quad
$\I{Z(\vt)}\mu=\mu(Z)(\I\vt\mu)$
\end{center}



{\bf Higher-order size algebra.} In the first-order case, the
interpretation of the function symbols $\ff$ such that $\ff^\cA$ is
not the constant function equal to $\infty$ (which includes
constructors) is $\ff^\cM(\va)=\I{\ff^\cA(\valpha)}\mu$ where
$\valpha\mu=\va$. To be able to do the same thing in the higher-order
case, we need the size algebra $\cA$ to be a typed higher-order
algebra interpreted in the sets $N_T$.

Hence, now, we assume that size expressions are simply-typed
$\l$-terms over a typed signature $\S$, and that every function symbol
$\ff:\tf$ is interpreted by $\infty$ or a size expression
$\ff^\cA:\tf$. We then let $\s:\cT\a\o\cA$ be the function that
replaces in a term every symbol $\ff$ by $\ff^\cA$, all the terms
containing $\infty$ being identified. Hence, for each term $t$
containing no symbol $\ff$ such that $\ff^\cA=\infty$, we have
$\I{t}\mu=\I{\s(t)}\mu$. Finally, we define $\lts$ as the relation
such that $a\lts b$ if, for each $\mu$, $\I{a}\mu<_\kA\I{b}\mu$.

For instance, for a strictly-positive constructor $\fc:\vT\A\fB$ with
$T_i=\vU_i\A\fB_i$, we can assume that there is a symbol
$\fc^\cA\in\S$ interpreted by the function
$\fc^\kA(\vx)=sup\{x_i\vy_i\mid i\in\ind(\fc),\vy_i\in
N_\ps{\vU_i}\}+1$. Hence, with Brouwer's ordinals, we have $\s(\lim
f)=\lim^\cA f\gts\s(fn)=fn$.

Thus, using such an higher-order size algebra, we can conclude:

\begin{theorem}
\label{thm-ho}
SB-termination implies termination if constructors are
strictly-positive and the conditions of Theorems {}\ref{thm-fo} and
{}\ref{thm-sbt-plus-plus} are satisfied.
\end{theorem}

\begin{proof}
The proof is similar to the first-order case (Theorem
{}\ref{thm-fo}). We only point out the main differences.

We first check that $\cM$ is a quasi-model. The case of $\ps\b$ is
detailed in Appendix \version{G}{\ref{sec-beta}}. For $\ps\cR$, we use
the facts that $!^\cM_B(\G)(l\t)\le_{M_B(\G)}!^\cM_B(\G)(r\t)$ if
$!^\cM_B(\G)(l\t)(\va)\le_{M_B(\vide)}!^\cM_B(\G)(r\t)(\va)$ for each
$\va$, and that $!^\cM_B(\G)(l\t)(\va)=\I{l}\t\mu$ where $x_i\mu=a_i$.

We do not label applications and abstractions. And for a defined
symbol $\ff:\vB\A B$, we take
$S^\ff=\coprod_\G\prod_{i=1}^nM_{B_i}(\G)$ and
$\pi_\ff(\G)(\vx)=(\G,\vx)$.

We now define a well-founded relation on $S^\ff$ that we will use for
proving some higher-order version of precedence-termination. For
dealing with $lab(\ps\cR)$, let $(\G,\vx)>_\ff^\cR(\D,\vy)$ if
$\D=\G+\G'$ and, for each $\vz\vz'$, $\z_\ff(\ldots
x_i(\vz)\ldots)>_\ff^\kA\z_\ff(\ldots y_i(\vz\vz')\ldots)$. For
dealing with $lab(\ps\b)$, let $(\G,\vx)>_\ff^\b(\D,\vy)$ if $\G=\D+T$
and there is $e$ such that, for each $i$ and $\vz$,
$x_i(\vz,e(\vz))=y_i(\vz)$. Since $>_\ff^\cR\circ>_\ff^\b$ is included
in $>_\ff^\cR\cup>_\ff^\b\circ>_\ff^\cR$, the relation
${>_\ff}={>_\ff^\cR\cup >_\ff^\b}$ is well-founded
{}\cite{dornbos98igpl}.

One can easily check that the functions $\pi_\ff$ and $\ff^\cM$ are
monotonic.

We are now left to prove that $\a_{lab(\ps\b)\cup lab(\ps\cR)\cup
Decr}$ terminates. First, remark that $\a_{lab(\ps\b)}$ is included in
$\a_{Decr}^*\a_\ps\b$. Indeed, given $@_T^U(\l_T^U(\l
xlab_U(\G+T)(u)),lab_T(\G)(t))\a lab_U(\G)(u_x^t)\in lab(\ps\b)$, a
symbol $\ff$ occuring in $u$ is labelled in $lab_U(\G+T)(u)$ by
something like $(\G+T+\D,!^\cM_\vB(\G+T+\D)(\vv))$, and by something
like $(\G+\D,!^\cM_\vB(\G+\D)(\vv_x^t))$ in $lab_U(\G)(u_x^t)$. Hence,
the relation $\a_{lab(\ps\b)\cup lab(\ps\cR)\cup Decr}$ terminates if
$\a_{\ps\b\cup lab(\ps\cR)\cup Decr}$ terminates.

By translating back IDTS types to simple types and removing the
symbols $\l_T^U$ (function $|\u~|$), we get a {\em $\b$-IDTS}
{}\cite{blanqui00rta} such that $\a_{\ps\b\cup lab(\ps\cR)\cup Decr}$
terminates if $\a_{|\ps\b\cup lab(\ps\cR)\cup Decr|}$ terminates
(Appendix \version{F}{\ref{sec-idts}}). Moreover, after
{}\cite{blanqui00rta}, $\a_{|\ps\b\cup lab(\ps\cR)\cup Decr|}$
terminates if $|lab(\ps\cR)\cup Decr|$ satisfies the {\em General
  Schema} (we do not need the results on {\em solid} IDTSs
{}\cite{hamana07ppdp}). This can be easily checked by using the
precedence $>$ on $\o\cF$ such that $\ff_a>\fg_b$ if $\ff\gtf\fg$ or
$\ff\eqf\fg$ and $a>_\ff b$.
\end{proof}


{\bf Conclusion.} By studying the relationship between sized-types
based termination and semantic labelling, we arrived at a new way to
prove the correctness of SBT that enabled us to extend it to
non-constructor systems, \ie systems with matching on defined symbols
(\eg associative symbols, Appendix \version{D}{\ref{sec-non-constr}}).
This work can be carried on in various directions by considering:
richer type structures with polymorphic or dependent types,
non-strictly positive constructors,
or the inference of size annotations to automate SBT.

{\bf Acknowledgments.} The authors want to thank very much Colin Riba
and Andreas Abel for their useful remarks on a previous version of
this paper. This work was partly supported by the
Bayerisch-Franz\"osisches Hochschulzentrum.


\newpage
\appendix

\section{Pattern condition}
\label{sec-pat-cond}

Example of pattern not satisfying the pattern condition:

Consider the (higher-order) base type $\fB$ whose constructors are
$\fb:\fB^\alpha\A\fB^\alpha\A\fB^{\fs\alpha}$, $\fc:\fB^\alpha\A\fB^\alpha$,
$\fd:(\fN\A\fB^\alpha)\A\fB^{\fs\alpha}$ and $\fe:\fB^\infty$.

Because of the constructor $\fd$, $\I\fB$ has elements of size greater
then $\w$. For instance, $\fd(\fj)$ where $\fj:\fN\A\fB$ is defined by
the rules $\fj\mf{0}\a\fe$ and $\fj(\fs x)\a\fc(\fj x)$, is of size
$\w+1$.

Consider now the pattern $p=\fb x(\fc(\fc y))$ for a function
$\ff:\fB^\alpha\A T$.

Since the types of the arguments of a constructor use the same size
variable $\alpha$, for $p$ to be well typed, we need to take
$\G=x:\fB^{\fs(\fs\b)},y:\fB^\b$ and $a=\fs(\fs(\fs\b))$.

Hence, assuming that $p\in\I\fB$, there must be an ordinal $\kb=\b\nu$
such that $o(x)\le\kb+2$, $o(y)\le\kb$ and $\kb+3\le
o(p)=max\{o(x)+1,o(y)+3\}$. Unfortunately, if we take an element $x$
of size $o(x)=\w+1$ and an element $y$ of size $o(y)=0$, then the
previous set of constraints, which reduces to $\w+1\le\kb+2$ and
$\kb+3\le\w+2$, is unsatisfiable. Indeed, for being satisfiable, $\w$
should be a successor ordinal which is not the case.


\section{Type inference}
\label{sec-typ-inf}

Let $\cX(\G)=\bigcup_{x\in\dom(\G)}\cX(x\G)$ be the set of size
variables occuring in the types of the variables of $\dom(\G)$.

\begin{figure}
\vsp[-5mm]
\centering
\caption{Type inference system\label{fig-thi-mgu}}
\vsp[3mm]
$\cfrac{(x,T)\in\G}{\G\th^ix:T}$\quad\quad
$\cfrac{\r:\cX(\tf^\cA)\a\cX\moins\cX(\G)\mbox{ renaming}}
{\G\th^i\ff:\tf^\cA\r}$\quad\quad
$\cfrac{\G,x:T\th^iu:U}{\G\th^i\lx^Tu:T\A U}$\\[3mm]
$\cfrac{\begin{array}{c}
\G\th^it:U\A V\quad \G\th^iu:U'\\
\r:\cX(U')\moins\cX(\G)\a\cX\moins(\cX(U)\cup\cX(\G))\mbox{ renaming}\\
\vphi=mgu(U,U'\r)\mbox{ with }\cX(\G)\mbox{ seen as constants}
\end{array}}{\G\th^i tu:V\vphi}$
\vsp[-6mm]
\end{figure}

\begin{lemma}
\label{lem-typ-inf}
The type inference relation of Figure {}\ref{fig-thi-mgu} is correct
and complete wrt the typing relation of Figure {}\ref{fig-th} with the
subtyping rules removed:
\begin{itemize}
\item
If $\G\th^it:T$ then $\G\th t:T$.
\item
If $\G\th t:T$ then there is $T'$ and $\vphi$ such that $\G\th^it:T'$
and $T'\vphi=T$.
\end{itemize}
\end{lemma}

\begin{proof}
\begin{itemize}
\item Correctness:
By induction on $\G\th^it:T$, using stability by substitution.
\item Completeness:
By induction on $\G\th t:T$. We only detail the application case. By
induction hypothesis, there is $T'$ and $\vphi$, and $U'$ and $\psi$
such that $\G\th^it:T'$, $T'\vphi=U\A V$, $\G\th^iu:U'$ and
$U'\psi=U$. It follows that $T'$ is of the form $A\A B$ and $U=A\vphi$
and $B\vphi=V$. Hence, there is $\t=mgu(A,U')$ and $\vphi'$ such that
$\vphi=\t\vphi'$. Therefore, $\G\th^itu:B\t$ and there is $\vphi'$
such that $B\t\vphi'=V$.\qed
\end{itemize}
\end{proof}


\section{Proof of Theorem {}\ref{thm-sbt-plus}}
\label{sec-sbt-plus}

\begin{proof}
We prove that, for all $\t$, if $\vp\t\in\I\vP$ and
$\vl\t\in\I{\vec\fB}$ then there is $\nu$ such that, for all
$(x,T)\in\G$, $x\t\in\I{T}^\nu$ and $\va\nu=o_{\vec\fB}(\vl\t)$.

Let $T$ be a type in which $\pos(\alpha,T)\sle\pos^+(T)$. Then,
$\I{T}_\alpha^\ka$ is a monotonic function on $\ka$
{}\cite{blanqui04rta}. Given $t\in\I{T}_\alpha^\kA$, let $o_{\l\alpha
T}(t)$ be the smallest ordinal $\ka$ such that
$t\in\I{T}_\alpha^\ka$. Note that $o_{\l\alpha\fB^\alpha}=o_\fB$.

Let now $(x,T)\in\G$. Since we have an inductive structure,
$\pos(x,T)\sle\pos^+(T)$. One can easily check that $x\t\in\I{T}^\mu$
where $\mu$ is the constant valuation equal to $\kA$. We can thus
define $x\nu=o_{\l xT}(x\t)$ and we have $x\t\in\I{T}^\nu$.

We now prove that $a_i\nu=\s(l_i)\nu=o_{\vec\fB}(l_i\t)$ by induction
on $l_i$. If $l_i=x$ and $(x,T)\in\G$ then $\s(l_i)=x$, $T=\fB_i^x$
and $x\nu=o_{\l x\fB_i^x}(x\t)=o_{\fB_i}(l_i\t)$. Assume now that
$l_i=\fc\vt$ with $\fc:\vT\A\fC$. If $\fC$ is non-recursive, then
$\s(l_i)=0$ and $\s(l_i)\nu=0=o_{\fB_i}(l_i\t)$. Otherwise,
$\s(l_i)=\fs(\max(\s(t_{i_1}),\ldots,\s(t_{i_k})))$. If $T_{i_j}$ is a
base type then, by induction hypothesis,
$\s(t_{i_j})\nu=o_{T_{i_j}}(t_{i_j}\t)$. Otherwise, there is
$(x,T)\in\G$ such that $t_{i_j}=x$ and $\s(t_{i_j})\nu=x\nu=o_{\l
xT}(x\t)$. Since $o_\fC(l_i\t)=sup\{o_{\l\alpha\vT}(\vt\t)\}+1$, we
have $\s(l_i)\nu=o_{\fB_i}(l_i\t)$.\qed
\end{proof}


\section{Example of non-constructor system}
\label{sec-non-constr}

Assuming that $\fA$ is the $\A$-type constructor, then the expression
$\fF nuv$ defined below represents the set of $n$-ary functions from
$u$ to $v$.

\begin{center}
\begin{tabular}{cc}
\begin{rew}
+\mf{0}y & y\\
+(\fs x)y & \fs(+xy)\\
+(\fs x)y & +x(\fs y)\\
+(+xy)z & +x(+yz)\\
\end{rew}
&\quad
\begin{rew}
\fF\mf{0}uv & v\\
\fF(\fs x)uv & \fA u(\fF xuv)\\
\fF(+xy)uv & \fF xu(\fF yuv)\\
\end{rew}
\end{tabular}
\end{center}

Take $+^\cA(x,y)=\z_+(x,y)=a=2x+y+1$, $\fF^\cA=\infty$ and
$\z_\fF(x,u,v)=x$. The interpretation of $\fF^\cM$ is well-defined
since $x<a$ and $y<a$. The labelled system that we obtain (where
$b=2y+z+1$) is precedence-terminating:

\begin{center}
\begin{tabular}{cc}
\begin{rew}
+_{y+1}\mf{0}y & y\\
+_{a+2}(\fs x)y & \fs(+_axy)\\
+_{a+2}(\fs x)y & +_{a+1}x(\fs y)\\
+_{2a+z+1}(+_axy)z & +_{2x+b+1}x(+_byz)\\
\end{rew}
&
\begin{rew}
\fF_{0}\mf{0}uv & v\\
\fF_{x+1}(\fs x)uv & \fA u(\fF_xxuv)\\
\fF_a(+_axy)uv & \fF_xxu(\fF_yyuv)\\
\end{rew}
\end{tabular}
\end{center}


\section{$\cF$-monoids}
\label{sec-monoid}

To interpret (higher-order) substitutions, a presheaf $M$ must be an
{\em $\cF$-monoid}, \ie a monoid $(M,\mu:M^2\a M)$ compatible with the
structure of $\cF$-algebra:

\begin{itemize}
\item
$\mu_B(\G)(\io_B(\D)(i),\vu)=u_i$;
\item
$\mu_B(\G)(t,\io_{\D(1)}(\G)(1)\ldots\io_{\D(p)}(\G)(p))=t$;
\item
for $t\in M_B(\T)$, $u_i\in M_{\T(i)}(\D)$ and $v_i\in M_{\D(i)}(\G)$,\\
$\mu_B(\G)(\mu_B(\D)(t,\vu),\vv)=\mu_B(\G)(t,\mu_{\T(1)}(\G)(u_1,\vv)\ldots\mu_{\T(p)}(\G)(u_p,\vv))$;
\item
for $\ff:(\vB_1\A B_1)\A\ldots\A(\vB_n\A B_n)\A B$ and
$\G_i=\G+\vB_i$,\\ $\mu_B(\G)(\ff^\cM(\D)(\vt),\vu)
=\ff^\cM(\G)(\mu_{B_1}(\G_1)(t_1,\vv_1),\ldots,\mu_{B_n}(\G_n)(t_n,\vv_n))$\\
where $v_{i,j}=up_\G^{\vB_i}(u_j)$ if $j<|\D|$, and
$v_{i,j}=|\G|+j-|\D|$ otherwise.
\end{itemize}

In the category of $\cF$-monoids, the presheaf of meta-terms $I^\cZ$
equipped with the product $\mu_B(\G)(t,\vu)=t\{x_1\to
u_1,\ldots,x_n\to u_n\}$ (simultaneous substitution) is free. Hence,
given an $\cF$-monoid $M$, any valuation $\phi:\cZ\a M$ can be
uniquely extended into an $\cF$-monoid morphism $\phi^*:I^\cZ\a M$
such that:

\begin{itemize}
\item
$\phi^*_B(\G)(x)=\io_B(\G)(x)$;
\item
for $Z:\vB\A B$,\\
$\phi^*_B(\G)(Z(t_1,\ldots,t_n))
=\mu_B(\G)(\phi_B(\vB)(Z),\phi^*_{B_i}(\G)(t_1)\ldots\phi^*_{B_n}(\G)(t_n))$;
\item
for $\ff:(\vB_1\A B_1)\A\ldots\A(\vB_n\A B_n)\A B$ and
$\G_i=\G,\vx_i:\vB_i$,\\
$\phi^*_B(\G)(\ff(\l\vx_1t_1,\ldots,\l\vx_nt_n))
=(\ff^\cM)_B(\G)(\phi^*_{B_1}(\G_1)(t_1),\ldots,\phi^*_{B_n}(\G_n)(t_n))$.
\end{itemize}

Given a labelled term $t$, let $|t|$ be the term obtained after
removing all labels. The presheaf of labelled meta-terms $\o{I}^\cZ$
has a structure of $\cF$-monoid for each valuation $\t:\cZ\a I^\vide$
by taking:

\begin{itemize}
\item
$\mu^\t_B(\G)(i,\vu)=u_i$;
\item
for $Z:\vB\A B$,
$\mu^\t_B(\G)(Z(t_1,\ldots,t_n),\vu)=
Z(\mu^\t_{B_1}(\G)(t_1),\ldots,\mu^\t_{B_n}(\G)(t_n))$;
\item
for $\ff:(\vB_1\A B_1)\A\ldots\A(\vB_n\A B_n)\A B$, $\G_i=\G+\vB_i$
and $u_i\in\o{I}^{\t,\cZ}_{\D(i)}(\G)$,\\
$\mu^\t_B(\G)(\ff_a(\l\vx_1t_1,\ldots,\l\vx_nt_n),\vu)
=\ff_b(\mu^\t_{B_1}(\G_1)(t_1,\vv_1),\ldots,\mu^\t_{B_n}(\G_n)(t_n,\vv_n))$\\
where $b=\pi^\ff_B(\G)(!^\cM_{B_1}(\G_1)(|t_1|\t),\ldots,
!^\cM_{B_n}(\G_n)(|t_n|\t))$,\\ $v_{i,j}=up_\G^{\vB_i}(u_j)$ if
$j<|\D|$, and $v_{i,j}=|\G|+j-|\D|$ otherwise.
\end{itemize}


\section{Translation to IDTS and $\b$-IDTS}
\label{sec-idts}

For the translation $\ps~$ from $\l$-terms to second-order IDTS terms,
we have the following properties:

\begin{lemma}
\begin{itemize}
\item
For all $t$ and $\t$, $\ps{t\t}=\ps{t}\ps\t$.
\item
If $t\a_{\b\cup\cR}u$ then $\ps{t}\a_{\ps\b\cup\ps\cR}\ps{u}$.
\end{itemize}
\end{lemma}

We now introduce a translation from a {\em structural} IDTS $I$ having
base types in $\bB$ and some symbols $\l_T^U:(T\A U)\A Arr(T,U)$ for
all $T,U\in\bB$, to a {\em non-structural} IDTS $J$ having base types
in $\cB$ and no symbol $\l_T^U:(T\A U)\A Arr(T,U)$. The symbols of $J$
are all the symbols symbols $|\ff|:|T_1|\A\ldots\A|T_n|\A\fB$ such
that $\ff:T_1\A\ldots\A T_n\A\fB$ is a symbol of $I$ distinct from
some $\l_T^U$. A meta-term in $I^\cZ_T(\G)$ is then translated into a
meta-term in $|I|^\cZ_{|T|}(|\G|)$ as follows:

\begin{itemize}
\item
$|x|=x$,
\item
$|\ff(t_1,\ldots,t_n)|=|\ff|(|t_1|,\ldots,|t_n|)$,
\item
$\l_T^U(\l xu)=\l x|u|$,
\item
$|Z(t_1,\ldots,t_n)|=Z(|t_1|,\ldots,|t_n|)$.
\end{itemize}

Given a set $\cS$ of rules in $I$, let $|\cS|$ be the set of rules
$|l|\a|r|$ in $|I|$ such that $l\a r\in\cS$.

\begin{lemma}
\begin{itemize}
\item
For all $t$ and $\t$, $|t\t|=|t||\t|$.
\item
If $t\a_\cS u$ then $|t|\a_{|\cS|}|u|$.
\end{itemize}
\end{lemma}

Note that $@_T^U(\l_T^U(\l xZ(x)),X)$ is translated into $@_T^U(\l
xZ(x),X)$. Hence, if $I$ has symbols $@_T^U:Arr(T,U)\A T\A U$ and
rules $@_T^U(\l_T^U(\l xZ(x)),X)$, then $|I|$ is a {\em $\b$-IDTS} and
$\a_{\ps\b\cup\cS}$ terminates if $|\cS|$ satisfies the General Schema
{}\cite{blanqui00rta}.

\comment{
An IDTS term $t$ can be translated back into a simply-typed $\l$-term
$|t|$ as follows:

\begin{itemize}
\item
$|x|=x$,
\item
for $\ff:T_1\A\ldots\A T_n\A\fB$ and $U_i=T_{i+1}\A\ldots\A
T_n\A\fB$,\\ $|\l_\ps{T_{k+1}}^\ps{U_{k+1}}(\l
x_{k+1}\ldots\l_\ps{T_n}^\ps{U_n}(\l
x_n\ff(t_1,\ldots,t_k,x_{k+1},\ldots,x_n))...)|=\ff|t_1|\ldots|t_k|$,
\item
$|\l_\ps{T}^\ps{U}(\l xu)|=\l x^T|u|$,
\item
$|@_\ps{U}^\ps{V}(t,u)=|t||u|$.
\end{itemize}

\begin{lemma}
\begin{itemize}
\item
For all $t$ and $\t$, $|\ps{t}\t|=t|\t|$.
\item
If $t\a_{\ps\b\cup\ps\cR}u$ then $|t|\a_{\b\cup\cR}|u|$.
\end{itemize}
\end{lemma}
}


\section{Validity of $\b$}
\label{sec-beta}

Using the interpretation of $@_T^U$ and $\l_T^U$ in Section \ref{sec-ho}:

\begin{lemma}
\label{lem-beta-valid}
If $(M,\mu)$ is an $\cF$-monoid, then $\ps\b$ is valid in $M$.
\end{lemma}

\begin{proof}
Let $l$ and $r$ be the left and right hand-sides of the rule $\b_T^U$,
$\t:\cZ\a I^\vide$ and $\G$. Assume that $\t(Z)=\lx u$ and
$\t(X)=t$. Then, $l\t=@_T^U(\l_T^U(\lx u),t)$ and $r\t=u_x^t$, and
$!^\cM_U(\G)(l\t)=(@_T^U)^\cM(\G)(\ku,\kt)=\mu_U(\G)(\ku,\vp\kt)$ and
$!^\cM_U(\G)(r\t)=!^\cM_U(\G)(u_x^t)$, where $\ku=!^\cM_U(\G,x:T)(u)$
and $\kt=!^\cM_T(\G)(t)$. We now prove by induction on $u$ that, for
all $\G$, $L=\mu_U(\G)(\ku,\vp\kt)$ is equal to
$R=!^\cM_U(\G)(u_x^t)$.

\begin{itemize}
\item $u=x$.
Then, $u_x^t=t$, $\ku=p_{n+1}$ and $L=\kt=R$.
\item $u=\ff(\l\vx_1u_1,\ldots,\l\vx_pu_p)$
with $\ff:(\vB_1\A B_1)\A\ldots\A(\vB_p\A B_p)\A B$. Then,
$u_x^t=\ff(\l\vx_1u_1{}_x^t,\ldots,\l\vx_pu_p{}_x^t)$,
$R=\ff^\cM(\G)(\va)$ where $a_i=!^\cM_{B_i}(\G+\vB_i)(u_i{}_x^t)$, and
$\ku=\ff^\cM(\G+T)(\vu^*)$ where
$u_i^*=!^\cM_{B_i}(\G+T+\vB_i)(u_i)$. Since $M$ is an $\cF$-monoid,
$L=\ff^\cM(\G)(\vb)$ where
$b_i=\mu_{B_i}(\G+T+\vB_i)(u_i^*,\vv_i)$. And, by induction
hypothesis, we have $b_i=a_i$.\qed
\end{itemize}
\end{proof}

\begin{lemma}
\label{lem-monoid}
$(M,\mu)$ is an $\cF$-monoid if
$\ff^\cM(\G)(\vx)(\vy)=\ff^\cM(\vide)(x_1(\vy),\ldots,x_n(\vy))$.
\end{lemma}

\begin{proof}
Let $\ff:(\vB_1\A B_1)\A\ldots\A(\vB_n\A B_n)\A B$ and
$\G_i=\G+\vB_i$. We have to prove that
$L=\mu_B(\G)(\ff^\cM(\D)(\vt),\vu)$ is equal to
$R=\ff^\cM(\G)(\mu_{B_1}(\G_1)(t_1,\vv_1)$, \ldots,
$\mu_{B_n}(\G_n)(t_n,\vv_n))$, where $v_{i,j}=up_\G^{\vB_i}(u_j)$ if
$j<|\D|$, and $v_{i,j}=|\G|+j-|\D|$ otherwise.

Let $y_i\in N_{\G(i)}$. We have $L(\vy)=\ff^\cM(\D)(\vt)(\vu')$ where
$u_j'=u_j(\vy)$. Now, by assumption, $L(\vy)=\ff^\cM(\vide)(\va)$
where $a_i=t_i(\vu')$, and $R(\vy)=\ff^\cM(\vide)(\vb)$ where
$b_i=\mu_{B_1}(\G_1)(t_i,\vv_i)(\vy)=t_i(\vv_i')$ and
$v'_{i,j}=v_{i,j}(\vy)$. Hence, $L(\vy)=R(\vy)$ since
$v'_{i,j}=u_j(\vy)=u_j'$.\qed
\end{proof}


\comment{
\section{Example of non-constructor system}
\label{sec-non-constr}

Assuming that $\fA$ is the $\A$-type constructor, then the expression
$\fF nuv$ represents the set of $n$-ary functions from $u$ to $v$.

\begin{center}
\begin{tabular}{cc}
\begin{rew}
+\mf{0}y & y\\
+(\fs x)y & \fs(+xy)\\
+(\fs x)y & +x(\fs y)\\
+(+xy)z & +x(+yz)\\
\end{rew}
&\quad
\begin{rew}
\fF\mf{0}uv & v\\
\fF(\fs x)uv & \fA u(\fF xuv)\\
\fF(+xy)uv & \fF xu(\fF yuv)\\
\end{rew}
\end{tabular}
\end{center}

We take:

\begin{itemize}
\item
$+^\cA(x,y)=2x+y+1$ and $\z_+(x,y)=2x+y+1$
\item
$\fF^\cA=\infty$ and $\z_\fF(x,u,v)=x$
\end{itemize}

The interpretation of $\fF^\cM$ is well-defined since $2x+y+1\le a$
implies both $x<a$ and $y<a$, and the labelled system that we obtain
is precedence-terminating:

\begin{rewc}
+_{y+1}\mf{0}y & y\\
+_{2x+y+3}(\fs x)y & \fs(+_{2x+y+1})\\
+_{2x+y+3}(\fs x)y & +_{2x+y+2}x(\fs y)\\
+_{4x+2y+z+3}(+_{2x+y+1}xy)z & +_{2x+2y+z+2}x(+_{2y+z+1}yz)\\
\fF_{0}\mf{0}uv & v\\
\fF_{x+1}(\fs x)uv & \fA u(\fF_{x}xuv)\\
\fF_{2x+y+1}(+_{2x+y+1}xy)uv & \fF_{x}xu(\fF_{y}yuv)\\
\end{rewc}
}

\end{document}